\documentclass[lettersize,journal]{IEEEtran}

\normalsize
\usepackage{amssymb, amsmath, epsfig, booktabs, amsthm, epstopdf, color, cite, bm}
\usepackage{graphicx}
\usepackage{subfig,float,stfloats,flushend}
\usepackage{psfrag}
\usepackage{paralist}
\usepackage{color}
\usepackage{setspace}
\usepackage{mathrsfs}
\usepackage{xcolor}
\usepackage[normalem]{ulem} 

\allowdisplaybreaks[4]

\newtheorem{theorem}{Theorem}
 
\newtheorem{lemma}{Lemma}

\newtheorem{remark}{Remark}

\renewcommand{\footnoterule}{%
  \kern -3pt
  \hrule width 150pt height .2pt
  \kern 2pt
}

\begin{document}

\title{\huge On Rate-Splitting With Non-unique Decoding In Multi-cell Massive MIMO Systems}

\normalsize

\author{\normalsize Meysam~Shahrbaf~Motlagh,~Subhajit~Majhi,~Patrick~Mitran,~Hideki~Ochiai
\thanks{
This work was supported in part
by the Natural Sciences and Engineering Research Council of Canada and Cisco. 

Meysam Shahrbaf Motlagh and Patrick Mitran are with the University of Waterloo, ON, Canada. Subhajit Majhi contributed to this work while he was affiliated with the University of Waterloo, and is currently with Qualcomm India. Hideki Ochiai is with Yokohama National University, Yokohama, Japan. \newline
This paper was presented in part in \cite{bsc2021} and is based on the first author's PhD dissertation \cite{myThesis}.}}
%
%
%
%
\maketitle
\begin{abstract}
We consider the downlink of a multi-cell massive MIMO system suffering from asymptotic rate saturation due to pilot contamination. As opposed to treating pilot contamination interference as noise (TIN), we study the performance of decoding the pilot contamination interference. We model pilot-sharing users as an interference channel (IC) and study the performance of schemes that decode this interference \textit{partially} based on rate-splitting (RS), and compare the performance to schemes that decode the interference in its \emph{entirety} based on simultaneous unique decoding (SD) or non-unique decoding (SND). For RS, we non-uniquely decode each layer of the pilot contamination interference and use one common power splitting coefficient per IC. Additionally, we establish an achievable region for this RS scheme. Solving a maximum symmetric rate allocation problem based on linear programming (LP), we show that for zero-forcing (ZF) with spatially correlated/uncorrelated channels and with a practical number of BS antennas, RS achieves significantly higher spectral efficiencies than TIN, SD and SND. Furthermore, we numerically examine the impact of increasing the correlation of the channel across antennas, the number of users as well as the degree of shadow fading. In all cases, we show that RS maintains significant gain over TIN, SD and SND.
\end{abstract}
\IEEEpeerreviewmaketitle
\section{Introduction}
\IEEEPARstart{I}{n} massive multi-input multi-output (MIMO) systems operating in time-division duplex (TDD) mode, channel state information (CSI) is estimated using uplink orthogonal training pilot sequences. However, since the channel coherence interval is finite, only a limited number of orthogonal pilots are available for channel estimation. In practice, a common approach taken to cope with this limitation is to re-use a set of orthogonal pilot sequences in different cells; hence, creating pilot contamination interference. As the number of BS antennas $M$ becomes large, this coherent interference scales at the same rate as the desired signal in the asymptotic limit of $M \rightarrow \infty$. Hence, by treating this interference as noise, the per-user rate asymptotically converges to a finite limit, which in turn saturates the benefits of using more antennas \cite{bjornson2017book}.  

While fully decoding interference or fully treating it as noise are two common strategies for interference management, it is known that these extreme strategies are not always optimal. For instance, treating interference as noise (TIN) is only preferred when interference is weak \cite{geng2015optimality}, whereas simultaneous unique decoding (SD) is only preferred when interference is strong \cite{bandemer2011interference}. Furthermore, while the results of \cite{8913624} revealed that simultaneous non-unique decoding (SND) outperforms both of these schemes, due to the structure of the capacity region of SND as the union of multiple-access channel (MAC) capacity regions, the SND decoder is still \textit{effectively} faced with only two options: treating each interfering signal as noise or fully decoding it; although which one to choose can be dynamically adapted at each receiver based on the strength of each interfering signal. Therefore, the proposed schemes of \cite{8913624} do not have the flexibility to decode only part of the interference while treating the remaining part as noise. For instance, such flexible decoding can be obtained by the celebrated Han-Kobayashi (HK) scheme \cite{han1981new}, which provides the best known achievable performance for the two-user interference channel (IC).

\subsection{Contributions}
Motivated in part by the recent introduction of practical sliding-widow codes that can, in principle, achieve 
the HK inner bound for the two-user IC \cite{wang2020sliding} as well as their extension to practical 5G settings \cite{kim2016interference}, in this work we propose a \textit{partial} interference decoding scheme to combat pilot contamination. Specifically, a novel scheme based on rate-splitting (RS) and superposition coding is proposed that can be applied to an IC with an arbitrary number of users. In the proposed partial decoding strategy, all users' messages are partitioned into two independent layers, an inner and an outer layer, so that each receiver can \textit{partially} decode the pilot contamination interference, if advantageous to do so. By bridging the extreme strategies of fully decoding pilot contamination interference or fully treating it as noise, the partial interference decoding scheme of this paper achieves significantly higher spectral efficiencies (SE) compared to TIN, SD and SND for the same number of BS antennas. 

We summarize the major contributions of this paper as follows:

\begin{itemize}
\item Using a worst-case uncorrelated noise technique, we derive a general achievable rate lower bound for the downlink of a multi-cell massive MIMO system that applies joint decoding to each set of pilot-sharing users. We then specialize this general lower bound to the case of ZF with a spatially uncorrelated Rayleigh fading channel model.
\item Motivated by the well-known HK scheme for a two-user IC, we propose a novel partial interference decoding scheme based on rate-splitting for an $L$-user IC, simply referred to as RS scheme in the following, that uses only one common power splitting coefficient per IC. 
In addition, we establish an achievable region for this RS scheme using the non-unique decoding technique. Since the rates of individual layers for each set of pilot-sharing users need to be adjusted globally across the entire network, the proposed RS scheme is implemented in a centralized manner (e.g., with the help of a central entity). 
\item To compare the performance of the different schemes, we numerically study the performance of maximum symmetric rate allocation for both spatially correlated and uncorrelated Rayleigh fading channels. 
We show that the maximum symmetric rate allocation problem can be formulated in terms of multiple linear programming (LP) problems. Moreover, an achievable sub-region of the RS scheme is introduced that provides an achievable lower bound to the performance of RS. In all cases, we observe that by numerically optimizing the power splitting coefficient, the proposed RS scheme produces significantly larger SEs compared to \textit{all the other schemes} for a practical number of BS antennas $M$, and its performance gain improves by increasing $M$. 
\item The impact of increasing the correlation of the channel across antennas, the number of users and the degree of shadow fading are also numerically studied. It is observed that while increasing the number of users and shadow fading degrade the performance of TIN, SND and RS, the gains provided by SND and RS over TIN increase; thus, demonstrating the importance of these schemes in practical settings. In addition, increasing the correlation magnitude of the channel improves the performance of TIN, SND and RS. Nonetheless, we observe that with a practical value of $M$, in all scenarios RS provides a significant gain over TIN, SD and SND.
\item Lastly, we show that by replacing the numerically-optimized value of the power splitting coefficient with its pre-computed average value, the performance loss is quite negligible; thus reducing the search space of the optimization problem in practical settings.
\end{itemize}
Note that the work of \cite{8913624} proposed interference decoding schemes based on SD and SND only for the uplink of a multi-cell massive MIMO system, assuming maximum ratio combining (MRC) and an uncorrelated Rayleigh fading channel model. In this work, the downlink setting is considered using zero-forcing (ZF) with spatially correlated/uncorrelated Rayleigh fading channel models, and a novel RS scheme using non-unique decoding is proposed and compared to SD and SND. Specifically, it is revealed that the use of ZF with RS provides significantly larger rates than TIN, SD and SND. Consequently, compared to the results of \cite{8913624}, the number of BS antennas required to outperform TIN is reduced by more than a factor of $\approx 100$.
\subsection{Related Work}\label{sec:1}
The RS technique was first introduced by Carleial \cite{carleial1978interference} for a two-user IC and was later used in the seminal work of \cite{han1981new} to establish the best known achievable performance for a two-user IC, which contains all other known schemes as special cases (e.g., joint decoding or TIN). The idea of splitting users' messages in conjunction with superposition coding has been widely adopted in the literature for interference mitigation in cellular networks \cite{dahrouj2011multicell, sahin2011interference, che2015successive, medra2018robust}. The work of \cite{dahrouj2011multicell} proposed an RS-based scheme in the downlink of a multi-cell network with perfect CSI to jointly design beamforming vectors for public and private parts. Therein, it was shown that by doing single-user successive decoding with a fixed decoding order, higher rates are achieved by this RS scheme compared to conventional TIN. Motivated by the HK scheme, \cite{sahin2011interference} proposed an interference cancellation technique via message splitting at the transmitter along with successive interference cancellation (SIC) decoding at the receiver that maximizes the sum-rate in heterogeneous networks. A similar technique has been adopted in \cite{che2015successive} to mitigate inter-cell interference in a multi-cell multi-user MIMO interference network. In another line of work, the idea of message splitting has been used to enhance the efficiency of medium access techniques. For instance, the work of \cite{mao2018rate} has proposed a rate-splitting multiple access technique that improves upon the performance of schemes such as space-division multiple access and non-orthogonal multiple access. 

Recently, RS has also been utilized in the context of massive MIMO communications with imperfect CSI \cite{dai2016rate,papazafeiropoulos2017rate,thomas2020rate}. Specifically, a novel hierarchical RS scheme is proposed in \cite{dai2016rate} for the downlink of a single-cell massive MIMO system operating in frequency-division duplex (FDD) mode. Therein, the precoding vector of each public part is designed so as to maximize the minimum rate of the public part achieved by each user. In \cite{papazafeiropoulos2017rate}, the benefits of RS are investigated to tackle adverse effects of hardware impairments in the downlink of a TDD-based massive multi-input single-output broadcast channel. Lastly, the work of \cite{thomas2020rate} has addressed the pilot contamination problem in a single-cell massive MIMO system operating in TDD mode, where all users inside the cell share the \textit{same} pilot sequence. While the authors have shown that the decoding scheme of \cite{thomas2020rate} achieves higher sum SE compared to the case without RS, by applying a single-user SIC decoder the intra-cell interference is still treated as noise.
  
\textit{Notation}: Boldface upper and lower case symbols are used to represent matrices and column vectors, respectively. The all-zero vector and $M \times M$ identity matrix are denoted by $\pmb{0}$ and $\pmb{I}_M$, respectively. $\mathcal{CN}(\pmb{m}, \pmb{R})$ denotes the circular symmetric complex Gaussian distribution with mean vector $\pmb{m}$ and covariance matrix $\pmb{R}$. The superscripts $(.)^T$ and $(.)^{\dag}$ denote transpose and the Hermitian transpose, respectively. The Shannon rate function is denoted by $C (x) = \log (1 + x)$, and $I (x; y)$ represents the mutual information between two random variables $x$ and $y$.
\section{Preliminaries}\label{sec:2}   
In the following, we first present the cell configuration, 
user placement 
and the channel propagation model in Section~\ref{sec:3}. 
This is followed by the downlink data transmission model in Section~\ref{sub:2}. Finally, in Section~\ref{sec:csi}, the user pilot assignment and pilot-based channel estimation are presented.     
\subsection{System Model}\label{sec:3}
A cellular network comprising $L$ cells, each having one BS located at the cell center and equipped with $M$ antennas serving $K$ single-antenna users, where $M > K$, is considered. Assuming a spatially correlated channel model, the channel matrix between BS $j$ and the users of cell $l$ is denoted by $\pmb{G}_{jl}=[\pmb{g}_{j1l},\pmb{g}_{j2l}, ..., \pmb{g}_{jKl}] \in \mathbb{C}^{M\times K}$. More precisely, the channel vector $\pmb{g}_{jkl} \in \mathbb{C}^{M\times 1}$ associated with user $k$ in cell $l$ is described by $\pmb{g}_{jkl} = \pmb{R}_{jkl}^{1/2} \pmb{h}_{jkl},$ where $\pmb{h}_{jkl} \sim \mathcal{CN} (\pmb{0}, \pmb{I}_M)$, and $\pmb{R}_{jkl} \in \mathbb{C}^{M\times M}$ is the spatial correlation matrix of the channel, i.e., $\pmb{g}_{jkl} \sim \mathcal{CN} (\pmb{0}, \; \pmb{R}_{jkl})$. 
A standard block-fading model is considered where the channels are constant over one coherence interval with one independent realization in each coherence block. Furthermore, followed by the TDD assumption, the uplink and downlink channels are reciprocal.
\subsection{Data Transmission Model}\label{sub:2}
During downlink data transmission, the $i^{\textrm{th}}$ user in cell $l$ receives the baseband signal
\begin{equation}\label{eq:3}
y_{il} = \sum_{j=1}^L \sqrt{\rho_{\textrm{dl}}} \pmb{g}_{jil}^{\dag} \pmb{x}_j + z_{il} , 
\end{equation}
where $\pmb{x}_j = \left[ x_j[1], ..., x_j[M] \right]^T $ is the transmit signal of BS $j$, $\rho_{\textrm{dl}}$ is the per-user transmit power of the BS, and $z_{il} \sim \mathcal{CN} (0, \; 1)$ is the receiver noise. Thus, $\rho_{\textrm{dl}}$ can also be interpreted as the per-user transmit signal-to-noise ratio (SNR) of the BSs. Also, defining $\pmb{w}_{jkj} \in \mathbb{C}^{M \times 1}$ as the precoding vector of user $k$ at BS $j$, we have 
\begin{align}\label{eq:3:0}
\pmb{x}_j &= \frac{1}{\sqrt{\lambda_j}}\sum_{k=1}^K \pmb{w}_{jkj} s_j [k] = \frac{\pmb{W}_{jj} \pmb{s}_j }{\sqrt{\lambda_j}}, 
\end{align}
where $\pmb{s}_j = \left[ s_j [1], ..., s_j [K] \right]^T $ is the vector of data symbols intended for the $K$ users in cell $j$, $\pmb{W}_{jj} = \left[ \pmb{w}_{j1j}, ..., \pmb{w}_{jKj} \right] \in \mathbb{C}^{M \times K}$, and $\lambda_j$ is a normalization factor to make sure the power constraint $\mathbb{E} [  \pmb{x}_j^{\dag} \pmb{x}_j ] / K = 1$ is satisfied at BS $j$, i.e., the downlink per-user transmit power of BS $j$ equals $\rho_{\textrm{dl}}$.
\subsection{CSI Estimation}\label{sec:csi}
Following \cite{adhikary2017uplink,ashikhmin2018interference}, we assume that the same pilot matrix $\pmb{\Psi} =[ \pmb{\psi}_1, \pmb{\psi}_2, ..., \pmb{\psi}_K ] ^T \in \mathbb{C}^{K \times K}$ incorporating orthonormal pilot sequences of length $K$, i.e., $\pmb{\Psi} \pmb{\Psi}^{\dagger} = \pmb{I}_K$, is shared in all cells across network, thus leading to pilot contamination. During the training stage, user $k \in \lbrace 1, ..., K \rbrace$ in cell $j$ transmits pilot $\pmb{\psi}_k$ to BS $j, j=1, ..., L$. After receiving the transmitted pilots, BS $j$ computes the observation signal 
\begin{equation}\label{eq:5}
\pmb{r}_{jk} = \sum_{l=1}^L \sqrt{\rho_{\textrm{p}}} \pmb{g}_{jkl} + \pmb{\tilde{z}}_{jk}, 
\end{equation} 
where $\pmb{\tilde{z}}_{jk} \sim \mathcal{CN} (\pmb{0}, \; \pmb{I}_M)$ and $\rho_{\textrm{p}}$ can be interpreted as the pilot SNR, which is generally a function of the uplink transmit power of users and the length of 
the pilot sequences (i.e., $K$). Therefore, the minimum mean-squared error (MMSE) estimate $\pmb{\hat{g}}_{jkj}$ of $\pmb{g}_{jkj}$, based on the observation $\pmb{r}_{jk}$ is found as 
\begin{align}
\nonumber \pmb{\hat{g}}_{jkj} &= \sqrt{\rho_{\rm p}} \pmb{R}_{jkj} \mathbb{E} \left[ \pmb{r}_{jk} \pmb{r}_{jk}^{\dagger} \right]^{-1} \pmb{r}_{jk} \\
&= \sqrt{\rho_{\rm p}} \pmb{R}_{jkj} \left( \sum\nolimits_{l=1}^{L} \rho_{\rm p} \pmb{R}_{jkl} + \pmb{I}_M \right)^{-1} \pmb{r}_{jk} \label{eq:6} \\
\nonumber &= \sqrt{\rho_{\rm p}} \pmb{R}_{jkj} \left( \sum\nolimits_{l=1}^{L} \rho_{\rm p} \pmb{R}_{jkl} + \pmb{I}_M \right)^{-1} \\
 &\hspace{5mm}\times\left( \underbrace{\sqrt{\rho_{\textrm{p}}} \pmb{g}_{jkj}}_{\textrm{Intended channel}} + \underbrace{ \sum\nolimits_{l=1, l\neq j}^L \sqrt{\rho_{\textrm{p}}} \pmb{g}_{jkl}}_{\textrm{Pilot contamination}} + \pmb{\tilde{z}}_{jk} \right), \label{eq:6:1}
\end{align}
i.e., the estimate of $\pmb{g}_{jkj}$ is contaminated by the channel of users in other cells sharing the same pilot sequence as user $k$ in cell $j$. Following the orthogonality property of MMSE estimation, $\pmb{g}_{jkj}$ can be decomposed as $\pmb{g}_{jkj} = \pmb{\hat{g}}_{jkj} + \pmb{\epsilon}_{jkj}$, where $\pmb{\epsilon}_{jkj}$ is the uncorrelated estimation error. From \eqref{eq:6:1}, it follows that the distribution of the channel estimate $\pmb{\hat{g}}_{jkj}$ and the estimation error $\pmb{\epsilon}_{jkj}$ are 
\begin{align}\label{eq:dist}
\pmb{\hat{g}}_{jkj} &\sim \mathcal{CN} \left(\pmb{0}, \; \rho_{\rm p} \pmb{R}_{jkj} \left( \sum_{l=1}^{L} \rho_{\rm p} \pmb{R}_{jkl} + \pmb{I}_M \right)^{-1} \pmb{R}_{jkj} \right), \\
\pmb{\epsilon}_{jkj} &\sim \mathcal{CN} \left(\pmb{0}, \; \hspace{-0.5mm}\pmb{R}_{jkj} - \hspace{-0.5mm} \rho_{\rm p} \pmb{R}_{jkj} \hspace{-0.5mm} \left( \sum_{l=1}^{L} \rho_{\rm p} \pmb{R}_{jkl} + \pmb{I}_M \right)^{\hspace{-1mm} -1} \hspace{-1mm} \pmb{R}_{jkj}  \hspace{-0.5mm} \right). \label{eq:dist:1}  
\end{align}
Moreover, for the estimate of $\pmb{g}_{jkl}, l \neq j$, i.e., channels of users in other cells sharing the same pilot sequence $\pmb{\psi}_k$, provided that $\pmb{R}_{jkj}$ is invertible, using \eqref{eq:6} we have 
\begin{align}
\nonumber \pmb{\hat{g}}_{jkl} &= \sqrt{\rho_{\rm p}} \pmb{R}_{jkl} \left( \sum\nolimits_{l=1}^{L} \rho_{\rm p} \pmb{R}_{jkl} + \pmb{I}_M \right)^{-1} \pmb{r}_{jk} \\  
&= \pmb{R}_{jkl} (\pmb{R}_{jkj})^{-1} \pmb{\hat{g}}_{jkj}. \label{eq:estimate:ratio} 
\end{align}
\section{Decoding Interference Fully}
\label{sec:4}
\begin{figure}[t] 
\centering
\includegraphics[scale=0.5]{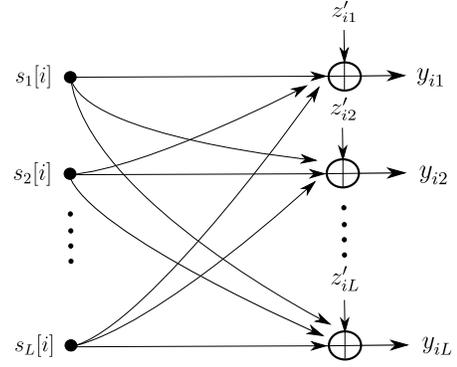}
\caption{The $L$-user IC associated with the $i^{\rm th}$ user of each cell, sharing pilot sequence $\pmb{\psi}_i$. There are a total of $K$ non-interfering such ICs in the network.\label{downlink:ic}} 
\end{figure}
In downlink, the received signal at user $i$ in cell $l$ can be re-written as 
\begin{align} 
\nonumber y_{il} &= \underbrace{ \sum_{j=1}^L  \sqrt{\dfrac{\rho_{\rm dl}}{\lambda_j}} \mathbb{E} \left[ \pmb{g}_{jil}^{\dagger} \pmb{w}_{jij} \right] s_j[i]}_{\textrm{Signal components}} \\
\nonumber &\quad + \underbrace{ \sum_{j=1}^L \sqrt{\dfrac{\rho_{\rm dl}}{\lambda_j}} \left( \pmb{g}_{jil}^{\dag} \pmb{w}_{jij} - \mathbb{E} \left[ \pmb{g}_{jil}^{\dag} \pmb{w}_{jij} \right] \right) s_j [i] }_{\textrm{Interference due to beamforming gain uncertainty}}  \\
	&\quad + \underbrace{ \sum\nolimits_{j=1}^L \sqrt{\dfrac{\rho_{\rm dl}}{\lambda_j}} \sum\nolimits_{k=1, k \neq i}^K \pmb{g}_{jil}^{\dag} \pmb{w}_{jkj} s_j [k]}_{\textrm{Inter-cell interference caused by other users}} + \underbrace{ z_{il}}_{\textrm{Noise}} \label{eq:down:1} \\
	&=   \underbrace{ \sum_{j=1}^L \zeta_{ij} s_j [i]}_{\textrm {Signal components}} + \underbrace{z_{il}^{\prime}}_{\textrm{Additive noise}}, \label{eq:down:2}  
\end{align}
where the signals of the $i^{\rm th}$ pilot-sharing users are now treated as desired signal components, $\zeta_{ij} := \sqrt{{\rho_{\rm dl}}/{\lambda_l}} \mathbb{E} [ \pmb{g}_{jil}^{\dag} \pmb{w}_{jij} ] $ is the effective channel gain and $z_{il}^{\prime}$ is the additive noise incorporating the last three terms of \eqref{eq:down:1}. By considering the received signals $y_{il}$, for $l=1, ..., L, \; i=1, ..., K$, together we have the outputs of $K$ non-interfering $L$-user ICs, with input signals $\lbrace s_j [i] \rbrace_{j=1}^L$, as in Fig.~\ref{downlink:ic}. Since the power of the terms inside the first sum in \eqref{eq:down:1} scale at the same rate as $M$ grows, one can jointly decode signals $\lbrace s_j [i] \rbrace_{j=1}^L$ in their \textit{entirety}, either uniquely as in SD or non-uniquely as in SND, and potentially eliminate the rate saturation problem as $M \rightarrow \infty$. In the following, these two interference decoding schemes will be briefly revisited for the downlink setting of this paper.
\subsection{Simultaneous Unique Decoding (SD)}\label{sec:5}
In SD, at the $i^{\rm th}$ user of cell $l$, it is assumed that each message $m_{ij} \in [ 1 : 2^{n R_{ij}} ],\; j=1, ..., L, $ (distributed uniformly) is encoded into a single codeword $\pmb{s}_j^{n} [i] (m_{ij})$ of length $n$, generated i.i.d. $\mathcal{CN} (0, 1)$. Thus, by jointly decoding all messages $\lbrace m_{ij} \rbrace_{j=1}^L$ uniquely at user $i$ of cell $l$, the decoding error probability approaches zero for $n \rightarrow \infty$ if the set of rate tuples $\left[ R_{i1}, ..., R_{iL} \right]$ is such that 
\begin{equation}\label{eq:down:3}
\sum_{j \in \Omega} R_{ij} \leq I \left(  y_{il} ; \; \pmb{s}_{\Omega} \; \Big\vert \; \pmb{s}_{\Omega^c}  \right), 
\end{equation}
for all $\Omega \subseteq \mathcal{L}$, where $\mathcal{L} := \lbrace 1, ..., L \rbrace$, $\Omega^c = \mathcal{L} \setminus  \Omega $ and $\pmb{s}_{\Omega}$ is the vector with entries $ s_j[i],\; j \in \Omega$, i.e., belongs to the capacity region of an $L$-user MAC. The set of rate vectors $\left[ R_{i1}, ..., R_{iL} \right]$ that satisfies the inequalities of \eqref{eq:down:3} defines the achievable region at user $i$ in cell $l$, denoted by $\mathcal{R}_{il}^{\rm SD}$. Finally, to obtain the achievable region network-wide for the $i^{\rm th}$ pilot-sharing users, one should take the intersection of achievable regions over all receivers, i.e., $\mathscr{R}_i^{\rm SD} = \bigcap_{l=1}^L \mathcal{R}_{il}^{\rm SD}.$
\subsection{Simultaneous Non-unique Decoding (SND)}\label{sec:6}
In SND, at the $i^{\rm th}$ user of cell $l$, the intended signal $s_l [i]$ is decoded uniquely while the pilot contamination interference signals $s_j [i], \; j \neq l$, are decoded non-uniquely \cite{bandemer2015optimal}. As such, incorrectly decoding the interfering signals does not incur any penalty, hence relaxing some of the rate constraints of SD. More precisely, user $i$ of cell $l$ finds the unique message $\hat{m}_{il}$ such that $ ( \hat{\pmb{s}}_l^{n} [i] (\hat{m}_{il})$, $\pmb{\hat{s}}_{\mathcal{L} \setminus \lbrace l \rbrace}^{n} [i] (m_{ij})$, $\pmb{y}_{il}^{n} )$ is jointly typical \textit{for some} $m_{ij}$, where $\pmb{\hat{s}}_{\mathcal{L} \setminus \lbrace l \rbrace}^{n} [i] (m_{ij})$ is the tuple of all codewords $\pmb{\hat{s}}_j^{n}[i](m_{ij})$ for $j \in \mathcal{L} \setminus \lbrace l \rbrace $. Applying the results of \cite{bandemer2015optimal} to the downlink IC associated with the $i^{\rm th}$ pilot-sharing users, the achievable region obtained by SND at user $i$ of cell $l$ is given by
\begin{equation}\label{eq:24}
\mathcal{R}_{il}^{\rm SND} = \bigcup_{\lbrace l \rbrace \subseteq \Omega \subseteq \mathcal{L}} \mathcal{R}_{\textrm{MAC}(\Omega,l)}^i,
\end{equation}
where $\mathcal{R}_{\textrm{MAC}(\Omega,l)}^i$ represents the achievable rate region obtained from uniquely jointly decoding signals $s_j [i], \; j \in \Omega$ at user $i$ in cell $l$. Therefore, $\mathcal{R}_{\textrm{MAC}(\Omega,l)}^i$ is described by the following system of inequalities
\begin{equation}\label{eq:24:24}
\sum_{j \in \omega} R_{ij} \leq I \left(  y_{il} ; \; \pmb{s}_{\omega} \; \Big\vert \; \pmb{s}_{\Omega \setminus \omega}  \right), \quad \forall \omega \subseteq \Omega,
\end{equation}
\begin{figure*}[t!]
\normalsize
\setcounter{equation}{13}
\begin{align}\label{eq:thm:2:1}
I \left(  y_{il} ; \; \pmb{s}_{\omega}  \Big\vert  \pmb{s}_{\Omega \setminus \omega} \right) &\geq C \left(  \frac{ \sum_{j \in \omega} \dfrac{\rho_{\rm dl}}{\lambda_j} \left\vert \mathbb{E} \left[ \pmb{g}_{jil}^{\dagger} \pmb{w}_{jij} \right] \right\vert^2 }{ \sum_{j=1}^L \sum_{k=1}^K \dfrac{\rho_{\rm dl}}{\lambda_j} \mathbb{E} \left[  \left\vert \pmb{g}_{jil}^{\dagger} \pmb{w}_{jkj} \right\vert^2 \right] - \sum_{j \in \Omega} \dfrac{\rho_{\rm dl}}{\lambda_j} \left\vert \mathbb{E} \left[ \pmb{g}_{jil}^{\dagger} \pmb{w}_{jij} \right] \right\vert^2 + 1} \right). 
\end{align}
\hrulefill
\end{figure*}
where $\pmb{s}_{\omega}$ is the vector with entries $ s_j[i],\; j \in \omega$. Note that $\Omega$ in \eqref{eq:24} must contain the index of the intended signal $s_l [i]$. Finally, the network-wide capacity region associated with the $i^{\rm th}$ pilot-sharing users is obtained by intersecting the achievable regions over all receivers, i.e., $\mathscr{R}_i^{\rm SND} = \bigcap_{l=1}^L \mathcal{R}_{il}^{\rm SND}$. Note that, as the interfering signals are decoded \textit{non-uniquely} at each receiver, $\mathcal{R}_{\textrm{MAC} (\Omega,l)}^i$ is unbounded in rate coordinates $R_{ij}, \; j \in  \Omega^c $. Moreover, the signals $s_j [i], \; j \in \Omega^c$, are treated as noise in $\mathcal{R}_{\textrm{MAC} (\Omega,l)}^i$. As a consequence, it is readily verified that both TIN and SD regions are strictly contained in the SND region. Also, notice that each point inside the SND region is equivalent to the rate of decoding a subset of interfering users fully, while treating the remaining ones as noise. However, under SD, an interfering signal is \textit{always} fully decoded; hence performing poorly in the weak and moderate interference regimes. 
\begin{remark}\label{snd}
Note that practical schemes, known as sliding-window coded modulation and sliding-window superposition coding, have recently been introduced in the literature to decode interference in real-world cellular networks, and also achieve performance close to the SND region \cite{wang2014sliding, park2014interference, kim2015adaptive, kim2016interference, ahn2021apparatus, kim2017methods, kwangtaik2017method}.
\end{remark}
\subsection{Achievable Lower Bounds}\label{sec:7}
The achievable regions obtained by SD and SND in \eqref{eq:down:3} and \eqref{eq:24}, respectively, are described by conditional mutual information expressions $I(.;. \vert .)$ that should be characterized using \eqref{eq:down:2}. However, since the additive noise term $z_{il}^{\prime}$ in \eqref{eq:down:2} is neither Gaussian nor independent of the signal components, its exact computation is difficult. Nevertheless, it can be verified that $z_{il}^{\prime}$ is zero-mean and also uncorrelated from the signal components. Thus, using the worst-case uncorrelated noise technique for multi-user channels \cite[Lemma~1]{8913624}, one can obtain an achievable lower bound on these mutual information terms by replacing the additive noise $z_{il}^{\prime}$ with an independent zero-mean Gaussian noise having the same variance. This is formally stated in the following lemma.
\begin{lemma}\label{thm:2}
Assuming Gaussian signaling, i.e., $\left[ s_1 [i], s_2 [i], ..., s_L [i] \right]^T \sim \mathcal{CN} \left( \pmb{0}, \; \pmb{I}_L \right)$ for $i \in \lbrace 1, 2, ..., K \rbrace $, the set of achievable lower bounds shown in \eqref{eq:thm:2:1} at the top of the next page is obtained at user $i$ in cell $l$, 
for all $\omega \subseteq \Omega$, where $\Omega$ is such that $\lbrace l \rbrace \subseteq \Omega \subseteq \mathcal{L}$. Also, the expectations are taken with respect to the channel realizations, and $\lambda_j, j=1, ..., L,$ is found based on the choice of precoding vector $\pmb{w}_{jij}$.
\end{lemma}
\color{black}
\begin{proof}
See Appendix~A-1.
\end{proof}
Note that the numerator on the right side of \eqref{eq:thm:2:1} represents the received power of the subset $\omega \subseteq \Omega$ of the desired signal components, while in the denominator the first term is the received power of all signal components (i.e., first term in \eqref{eq:3}), the second term is the received power of all desired signal components (and these are subtracted from the first term), 
and finally the third term represents the noise power. 

Further note that this lower bound is valid in general and does not depend on specific choices of the precoding vector $\pmb{w}_{jij}$, and the expectations can be computed, if necessary, using Monte Carlo simulation. Furthermore, in the case of SD, one should use the substitution $\Omega = \mathcal{L}$ in the lower bound of \eqref{eq:thm:2:1}. In the following, we specialize the lower bound of \eqref{eq:thm:2:1} to ZF precoding with an uncorrelated Rayleigh fading channel model.

\emph{ZF with an uncorrelated channel:} Using ZF precoding, $\pmb{W}_{jj} = \pmb{\hat{G}}_{jj} ( \pmb{\hat{G}}_{jj}^{\dagger} \pmb{\hat{G}}_{jj})^{-1}$, we have $\pmb{W}_{jj}^{\dagger} \pmb{\hat{G}}_{jj} = \pmb{I}_K$, and thereby $\pmb{w}_{jkj}^{\dagger} \pmb{\hat{g}}_{jmj} = \delta_{mk}$, where $\delta_{mk}$ is the Kronecker delta function. We also assume that an uncorrelated Rayleigh fading channel model is used, i.e., $\pmb{R}_{jkl} = \beta_{jkl} \pmb{I}_M$, where $\beta_{jkl}$ is the large-scale fading coefficient between the antennas of BS $j$ and user $k$ in cell $l$, and is constant over many coherence intervals. Therefore, one can write
\begin{align}\label{eq:down:zf:a}
y_{il} &= \sum_{j=1}^L \sqrt{\dfrac{\rho_{\rm dl}}{\lambda_j^{\rm zf}}} \sum_{k=1}^K \left( \pmb{\hat{g}}_{jil}^{\dagger} + \pmb{\epsilon}_{jil}^{\dagger} \right) \pmb{w}_{jkj} s_j [k] + z_{il} \\
\nonumber &=  \underbrace{\sum_{j=1}^L \sqrt{\dfrac{\rho_{\rm dl}}{\lambda_j^{\rm zf}}} \pmb{\hat{g}}_{jil}^{\dagger} \pmb{w}_{jij} s_j [i]}_{\textrm{Signal components}} \nonumber \\
&\quad+ \underbrace{\sum_{j=1}^L \sqrt{\dfrac{\rho_{\rm dl}}{\lambda_j^{\rm zf}}} \sum_{k=1, k \neq i}^K  \pmb{\hat{g}}_{jil}^{\dagger} \pmb{w}_{jkj} s_j [k]}_{\textrm{Non-coherent interference caused by other users}} \nonumber \\
&\quad + \underbrace{\sum_{j=1}^L \sqrt{\dfrac{\rho_{\rm dl}}{\lambda_j^{\rm zf}}} \sum_{k=1}^K \pmb{\epsilon}_{jil}^{\dagger} \pmb{w}_{jkj} s_j[k]}_{\textrm{Iterference due to estimation error}} + \underbrace{z_{il}}_{\textrm{Noise}}  \nonumber  \\
\nonumber &\stackrel{\rm (a)}{=} \underbrace{\sum_{j=1}^L \sqrt{\dfrac{\rho_{\rm dl}}{\lambda_j^{\rm zf}}} \left( \dfrac{\beta_{jil}}{\beta_{jij}} \right)  s_j [i]}_{\textrm{Signal components}} \\
 &\hspace{4mm}+ \underbrace{\sum_{j=1}^L \sqrt{\dfrac{\rho_{\rm dl}}{\lambda_j^{\rm zf}}} \sum_{k=1}^K \pmb{\epsilon}_{jil}^{\dagger} \pmb{w}_{jkj} s_j [k] + z_{il}}_{\textrm{ Interference + Noise}} \label{eq:down:zf:b}\\
&= \underbrace{\sum_{j=1}^L \theta_{ij} s_j [i]}_{\textrm{Signal components}} + \underbrace{z_{il}^{\prime \prime}}_{\textrm{Additive noise}} , \label{eq:down:zf:c} 
\end{align}
where (a) follows from \eqref{eq:estimate:ratio} and the substitution $\pmb{R}_{jkl} = \beta_{jkl} \pmb{I}_M$ and by noting that $\pmb{w}_{jkj}^{\dagger} \pmb{\hat{g}}_{jmj} = \delta_{mk}$,
and $\lambda_j^{\rm zf}$ is chosen such that $\mathbb{E}[  \pmb{x}_j^{\dag} \pmb{x}_j ] / K = 1$.
Also, $\theta_{ij} := \sqrt{{\rho_{\rm dl}}/{\lambda_j^{\rm zf}}} \left( {\beta_{jil}}/{\beta_{jij}} \right)$, and $z_{il}^{\prime \prime}$ is the additive noise term which is neither Gaussian nor independent of the desired signal components. Nonetheless, $z_{il}^{\prime \prime}$ is zero-mean and uncorrelated from the desired signal components. Therefore, one can apply the worst-case uncorrelated noise technique and obtain the achievable lower bound presented in the following theorem. 
\begin{figure*}[t!]
\normalsize
\setcounter{equation}{17}
\begin{equation}\label{thm:down:zf:1}
I \left(  y_{il} ; \; \pmb{s}_{\omega}  \Big\vert  \pmb{s}_{\Omega \setminus \omega} \right) \geq C \left( \frac{\sum_{j \in \omega} \left( \dfrac{\rho_{\rm dl}}{\lambda_j^{\rm zf}} \right) \left( \dfrac{\beta_{jil}}{\beta_{jij}} \right)^2 }{ \sum_{j=1}^L \left( \dfrac{\rho_{\rm dl}}{\lambda_j^{\rm zf}} \right) \sum_{k=1}^K   \dfrac{\beta_{jil} - \sqrt{\rho_{\rm p}} \beta_{jil} \alpha_{jil}}{(M-K) \sqrt{\rho_{\rm p}} \beta_{jkj} \alpha_{jkj}} + \sum_{j \in \Omega^c} \left( \dfrac{\rho_{\rm dl}}{\lambda_j^{\rm zf}} \right) \left( \dfrac{\beta_{jil}}{\beta_{jij}} \right)^2  + 1 } \right) .
\end{equation} 
\hrulefill
\end{figure*}
\begin{theorem}\label{thm:down:zf} 
Assuming ZF and also Gaussian signaling, i.e., $ \left[ s_1 [i], s_2 [i], ..., s_L [i] \right]^T \sim \mathcal{CN} \left(\pmb{0}, \; \pmb{I}_L  \right)$, for $i \in \lbrace 1, 2, ..., K \rbrace $, the set of lower bounds shown in \eqref{thm:down:zf:1} at the top of the next page can be achieved
for all $\omega \subseteq \Omega$, where $\Omega$ is such that $\lbrace l \rbrace \subseteq \Omega \subseteq \mathcal{L}$. Also, we have $\lambda_j^{\rm zf} = \frac{1}{K (M - K)} \sum_{i=1}^K  \frac{1}{\sqrt{\rho_{\rm p}} \beta_{jij} \alpha_{jij} }$, where $\alpha_{jij} := \frac{\sqrt{\rho_{\textrm{p}}} \beta_{jij}}{1 + \rho_{\textrm{p}}\sum_{l=1}^L \beta_{jil}}$.
\end{theorem}
\begin{proof}
See Appendix~A-2.
\end{proof}
Note that the numerator on the right side of \eqref{thm:down:zf:1} is the received power of the subset $\omega \subseteq \Omega$ of the desired signal components, while in the denominator the first term represents the variance of the interference caused by the estimation error (i.e., the second term in \eqref{eq:down:zf:b}), the second term is the received power of the signal components not included in the set of desired signals (i.e., $s_j[i], j \in \Omega^{c}$) and the third term is the noise power. Further note that in the case of SD, the substitution $\Omega = \mathcal{L}$ is used in this lower bound.
\section{Decoding Interference Partially: Rate-splitting}
In the following, we first briefly discuss how partial interference decoding based on HK can be applied to a two-cell massive MIMO system. Then, motivated by this scheme, we propose a generalization that can be applied to more than two cells (i.e., an $L$-user IC where $L$ is arbitrary). Recall that there are $K$ non-interfering ICs in the network and the analysis is thus the same with respect to the index $i \in \left\lbrace 1, ..., K  \right\rbrace$ of the users sharing pilot sequence $\pmb{\psi}_i$. As such, for the rest of this paper, to simplify notation, the index $i$ will be suppressed.
\subsection{Two Cells: HK}\label{sec:hk}
Consider the two-user IC of \eqref{eq:down:2} associated with the two users sharing the same pilot sequence in each cell in the downlink of a two-cell massive MIMO system (i.e., $L=2$). To achieve the HK inner bound for this IC, one can follow the (simple) scheme of \cite[Section~6.5]{el2011network}. As shown in \cite{chong2008han}, this scheme achieves the same inner bound as the original HK scheme \cite{han1981new}. The HK scheme proceeds as follows: 

\textbf{Encoding:} Adopting the scheme of \cite{el2011network} for the Gaussian case, at BS $l, l=1, 2$, message $m_{l}$ is first partitioned into two independent parts $m_{l}^{(a)} \in [ 1 : 2^{n R_{l}^{(a)}} ]$ and $m_{l}^{(b)} \in [ 1 : 2^{n R_{l}^{ (b)}} ]$ such that $R_{l} = R_{l}^{(a)} + R_{l}^{(b)}$. Then, part $m_{l}^{(b)}$ is encoded into codeword $\pmb{s}_l^{(b)}(m_{l}^{(b)})$ of length $n$ (known as the ``cloud center" which carries ``coarse information"), while part $m_{l}^{(a)}$ is encoded into another codeword $\pmb{s}_l^{(a)}(m_{l}^{(a)}, m_{l}^{(b)})$ of length $n$; finally the latter codeword is superimposed (or layered) on the former to produce a single codeword for transmission $\pmb{s}_l(m_{l}^{(a)}, m_{l}^{(b)}) = \pmb{s}_l^{(b)}(m_{l}^{(b)}) + \pmb{s}_l^{(a)}(m_{l}^{(a)},m_{l}^{(b)})$ (known as the ``satellite codeword" which carries the full information). The total transmit power budget at each BS is split into two fixed parts according to the power splitting coefficient $\mu_{l}  \in [0, 1],\; l=1, 2$: the fraction $\mu_{l}$ of the budget is allocated to the ``outer'' layer $\pmb{s}_l^{(a)},\; l=1, 2$, while the fraction  $(1- \mu_{l})$ of the budget is allotted to the ``inner'' layer $\pmb{s}_l^{(b)},\; l=1, 2$. Finally, $\pmb{s}_l^{(a)}$ and $\pmb{s}_l^{(b)}$ are chosen to be i.i.d., zero-mean circularly symmetric complex Gaussian, with powers determined by $\mu_{l}$ for $l=1, 2$.

\textbf{Decoding:} The user of cell $l, l=1, 2,$ decodes both the inner and outer layers of the intended message $(m_{l}^{(a)},m_{l}^{(b)})$ uniquely, and tries to \textit{non-uniquely} decode the inner layer of the interfering message $m_{j}^{(b)}, j \neq l$, while treating the outer layer $m_{j}^{(a)}, j \neq l,$ as noise. 

It can be verified that the regions, TIN, SD and SND are obtained as special cases of HK, when specific choices of $\mu_{1}$ and $\mu_{2}$ are picked at each BS, i.e., $\lbrace 0 \rbrace$ or $\lbrace 1 \rbrace$. The performance of HK for a two-cell massive MIMO system was studied in \cite{bsc2021}, where it was shown that by numerically optimizing $\mu_1$ and $\mu_2$, this partial decoding scheme significantly outperforms TIN, SD and SND for a practical number of antennas. 
\subsection{Beyond Two Cells: RS}\label{sec:beyond:two}
When going beyond two cells, one possible generalization of the HK scheme can be obtained by considering one power splitting coefficient for each user in the corresponding IC, i.e., $L$ different coefficients $\mu_{l} \in \left[0, 1  \right], l=1, ..., L,$ for the $L$-user IC of Fig.~\ref{downlink:ic}. However, taking the union over the combination of all such power splitting strategies seems infeasible, especially for networks with a large number of cells. This motivates the need for a more feasible generalization and the use of a much simpler power splitting strategy. In the following, we propose one possible application of partial decoding to more than two cells that uses only one common power splitting coefficient per IC. Furthermore, as opposed to HK, each receiver non-uniquely decodes \textit{both} layers of all pilot contamination interference terms. We show that by doing so, the proposed scheme can outperform TIN, SD and SND for a practical number of antennas. 

\textbf{Encoding:} Encoding is similar to the case of a two-cell system in Section~\ref{sec:hk} that uses superposition coding, except that now only one common power splitting coefficient $\mu$ is utilized by all users of the IC in Fig.~\ref{downlink:ic}. In particular, message $m_{l}, l=1, ..., L,$ is first partitioned into two independent parts $m_{l}^{(a)} \in [ 1 : 2^{n R_{l}^{ (a)}} ]$ and $m_{l}^{(b)} \in [ 1 : 2^{n R_{l}^{ (b)}} ]$ such that $R_{l} = R_{l}^{(a)} + R_{l}^{ (b)}, l=1, ..., L$. Next, part $m_{l}^{(b)}$ is encoded into codeword $\pmb{s}_l^{(b)}(m_{l}^{(b)}), l=1, ..., L,$ of length $n$, while part $m_{l}^{(a)}$ is encoded into another codeword $\pmb{s}_l^{(a)} (m_{l}^{(a)}, m_{l}^{(b)}), l=1, ..., L,$ of length $n$. Finally, the latter codeword is superimposed on the former to produce the satellite codeword for transmission $\pmb{s}_l (m_{l}^{(a)}, m_{l}^{(b)}) = \pmb{s}_l^{(b)} (m_{l}^{(b)}) + \pmb{s}_l^{(a)} (m_{l}^{(a)},m_{l}^{(b)}), l=1, ..., L$. Moreover, the total transmit power budget at all BSs is split into two fixed parts according to the power splitting coefficient $\mu  \in [0, 1]$: the fraction $\mu$ of the budget is allocated to the ``outer'' layer $\pmb{s}_l^{(a)} ,\; l=1, ..., L$, while the fraction  $(1- \mu)$ of the budget is allotted to the ``inner'' layer $\pmb{s}_l^{(b)} ,\; l=1, ..., L$. Lastly, $\pmb{s}_l^{(a)} $ and $\pmb{s}_l^{(b)} $ are chosen to be i.i.d., zero-mean circularly symmetric complex Gaussian, with powers determined by $\mu$.

\textbf{Decoding:} In the decoding stage, the SND scheme is applied to \textit{non-uniquely} decode both layers of all pilot contamination interference terms. Specifically, the decoder at receiver $l$ (i.e., user of cell $l$) uniquely decodes both the inner and outer layers of its own message $(m_{l}^{(a)},m_{l}^{(b)})$, and \textit{non-uniquely} decodes both layers of all interfering messages $\lbrace m_{j}^{(a)},m_{j}^{(b)} \rbrace, j \in \mathcal{L} \setminus \lbrace l \rbrace$. This is as opposed to HK, where only the inner layer of the interfering signal is decoded non-uniquely while treating the outer layer as noise.

Note that while full information $(m_{l}^{(a)},m_{l}^{(b)})$ is carried in the satellite codeword $\pmb{s}_l$, the inner layer $\pmb{s}_l^{(b)}$, only carries coarse information $m_{l}^{(b)}$. Therefore, due to the code construction, the inner layer can be decoded without decoding $m_{l}^{(a)}$ in the outer layer, whereas the outer layer can be decoded either jointly with the inner layer, i.e., $(m_{l}^{(a)},m_{l}^{(b)})$, or only after $m_{l}^{(b)}$ is decoded first in the inner layer.

In Appendix~A-3, a detailed derivation of the achievable region for a two-cell system is provided. Below, the general achievable region for $L \geq 2$ is presented. The achievability proof follows the same steps as that in Appendix~A-3, but is significantly more tedious. To characterize an achievable rate region for RS, we first need to define the following sets:
\begin{align}
\nonumber \mathcal{S}_l^{\rm RS} &:=  \mathcal{A}_1^{\rm RS} \times ... \times \mathcal{A}_{l-1}^{\rm RS} \times \left\lbrace \left\lbrace m_{l}^{(a)},m_{l}^{(b)} \right\rbrace  \right\rbrace \\ 
&\hspace{5mm}\times \mathcal{A}_{l+1}^{\rm RS} \times ... \times \mathcal{A}_L^{\rm RS} , \quad l=1, ..., L, \label{eq:sl}
\end{align} 
where $ \times $ denotes the Cartesian product and $\mathcal{A}_j^{\rm RS}$ is given by 
\begin{equation}\label{eq:Aj}
\mathcal{A}_j^{\rm RS} :=  \left\lbrace  \emptyset,  \left\lbrace m_{j}^{(b)} \right\rbrace, \left\lbrace m_{j}^{(a)}, m_{j}^{(b)} \right\rbrace \right\rbrace. 
\end{equation}
Furthermore, denote the achievable region for the rate vector $[ R_1^{(a)}, R_1^{(b)}, ..., R_L^{(a)}, R_L^{(b)} ]^T$ obtained by the proposed RS scheme at receiver $l$ and the network-wide achievable region by $\mathcal{R}_{l}^{\rm RS}$ and $\mathscr{R}^{\rm RS}$, respectively. Then, following the discussion in Appendix~A-3, we have 
\begin{equation}\label{ch4:eq:23}
\mathscr{R}^{\rm RS} = \bigcap_{l=1}^L \mathcal{R}_{l}^{\rm RS},
\end{equation} 
where
\begin{equation}\label{ch4:eq:24}
\mathcal{R}_{l}^{\rm RS} = \bigcup_{ \Omega_l \in \mathcal{S}_l^{\rm RS}  } \mathcal{R}_{\textrm{MAC}(\Omega_l, l ) }^{\rm RS}, 
\end{equation}
and $\mathcal{R}_{\textrm{MAC}(\Omega_l, l ) }^{\rm RS}$ is a modified MAC region (as will be explained in the following) obtained from jointly decoding the messages included in the set $\Omega_l$, where $\Omega_l$ is an element of $\mathcal{S}_l^{\rm RS}$ defined in \eqref{eq:sl}, and thus $\lbrace m_{l}^{(a)},m_{l}^{(b)} \rbrace \subseteq \Omega_l$. Also, note that messages not included in the set $\Omega_l$, are treated as noise in the region $\mathcal{R}_{\textrm{MAC}(\Omega_l, l ) }^{\rm RS}$ at receiver $l$. 

One should note that the $\vert \Omega_l \vert $-user MAC of $\mathcal{R}_{\textrm{MAC}(\Omega_l, l ) }^{\rm RS}$ has less than $2^{\vert \Omega_l \vert} -1$ constraints ($\vert \Omega_l \vert$ is the cardinality of the set $\Omega_l$), as some of the constraints will be removed because of the following. As shown in Appendix~A-3, if $\Omega_l$ contains messages of both layers $(m_{j}^{(a)},m_{j}^{(b)}),$ for some $j$, then those constraints that contain $R_{j}^{ (b)}$ but not $R_{j}^{(a)}$ will be removed from the rate region. In Appendix~A-3, for the case of a two-cell system we have explicitly identified these constraints at each receiver. Below, we provide the example of a three-cell system and discuss its achievable rate region with RS.

\textbf{Example ($L=3$):} The message of each transmitter is first partitioned into two independent parts: $m_{1}^{(a)}$ and $m_{1}^{(b)}$ at BS $1$, $m_{2}^{(a)}$ and $m_{2}^{(b)}$ at BS $2$, and $m_{3}^{(a)}$ and $m_{3}^{(b)}$ at BS $3$. Then, by applying superposition coding and non-unique decoding, the set $\mathcal{S}_l^{\rm RS}, l=1,2 ,3$, at each BS is given by:
\begin{align}
\nonumber \mathcal{S}_1^{\rm RS}    &= \left\lbrace \left\lbrace m_{1}^{(a)}, m_{1}^{(b)} \right\rbrace  \right\rbrace \times \left\lbrace \emptyset,  \left\lbrace m_{2}^{(b)} \right\rbrace, \left\lbrace m_{2}^{(a)}, m_{2}^{(b)} \right\rbrace  \right\rbrace  \\
&\hspace{5mm}\times   \left\lbrace \emptyset,  \left\lbrace m_{3}^{(b)} \right\rbrace, \left\lbrace m_{3}^{(a)}, m_{3}^{(b)} \right\rbrace  \right\rbrace  \label{rs:three:cell} \\
\nonumber \mathcal{S}_2^{\rm RS}    &=   \left\lbrace \emptyset,  \left\lbrace m_{1}^{(b)} \right\rbrace, \left\lbrace m_{1}^{(a)}, m_{1}^{(b)} \right\rbrace  \right\rbrace  \times \left\lbrace \left\lbrace m_{2}^{(a)}, m_{2}^{(b)} \right\rbrace \right\rbrace \\
&\hspace{5mm}\times  \left\lbrace \emptyset,  \left\lbrace m_{3}^{(b)} \right\rbrace, \left\lbrace m_{3}^{(a)}, m_{3}^{(b)} \right\rbrace  \right\rbrace   \\
\nonumber \mathcal{S}_3^{\rm RS}    &=  \left\lbrace \emptyset,  \left\lbrace m_{1}^{(b)} \right\rbrace, \left\lbrace m_{1}^{(a)}, m_{1}^{(b)} \right\rbrace  \right\rbrace  \\
&\hspace{5mm}\times   \left\lbrace \emptyset,  \left\lbrace m_{2}^{(b)} \right\rbrace, \left\lbrace m_{2}^{(a)}, m_{2}^{(b)} \right\rbrace  \right\rbrace  \times \left\lbrace \left\lbrace m_{3}^{(a)}, m_{3}^{(b)} \right\rbrace \right\rbrace.  \label{rs:three:cell:1} 
\end{align}
Therefore, the achievable region at each BS using the RS scheme with non-unique decoding is obtained by taking the union of $9$ modified MAC regions. For instance, at BS $1$, one needs to take the union of regions $\mathcal{R}_{\textrm{MAC}(\Omega_1, 1  ) }^{\rm RS}$ over the following $9$ elements of $\mathcal{S}_1^{\rm RS}$, denoted by $\Omega_1^{(j)}, j=1, ..., 9,$ 
\begin{align}
\nonumber \Omega_1^{(1)} &= \left\lbrace \left\lbrace  m_{1}^{(a)}, m_{1}^{(b)} \right\rbrace \right\rbrace, \\ 
\nonumber \Omega_1^{(2)} &= \left\lbrace  \left\lbrace m_{1}^{(a)}, m_{1}^{(b)} \right\rbrace , \left\lbrace m_{2}^{(b)}  \right\rbrace \right\rbrace, \\
\nonumber \Omega_1^{(3)} &= \left\lbrace \left\lbrace m_{1}^{(a)}, m_{1}^{(b)} \right\rbrace , \left\lbrace m_{3}^{(b)} \right\rbrace \right\rbrace , \\
\nonumber \Omega_1^{(4)} &= \left\lbrace \left\lbrace m_{1}^{(a)}, m_{1}^{(b)} \right\rbrace , \left\lbrace m_{2}^{(a)}, m_{2}^{(b)} \right\rbrace \right\rbrace, \\
\nonumber \Omega_1^{(5)} &= \left\lbrace  \left\lbrace m_{1}^{(a)}, m_{1}^{(b)} \right\rbrace , \left\lbrace m_{3}^{(a)}, m_{3}^{(b)} \right\rbrace \right\rbrace, \\
\nonumber \Omega_1^{(6)} &= \left\lbrace \left\lbrace m_{1}^{(a)}, m_{1}^{(b)} \right\rbrace , \left\lbrace m_{2}^{(b)} \right\rbrace , \left\lbrace m_{3}^{(b)} \right\rbrace \right\rbrace, \\
\nonumber \Omega_1^{(7)} &= \left\lbrace  \left\lbrace m_{1}^{(a)}, m_{1}^{(b)} \right\rbrace, \left\lbrace m_{2}^{(a)}, m_{2}^{(b)} \right\rbrace, \left\lbrace m_{3}^{(b)} \right\rbrace \right\rbrace, \\
\nonumber \Omega_1^{(8)} &= \left\lbrace \left\lbrace m_{1}^{(a)}, m_{1}^{(b)} \right\rbrace , \left\lbrace m_{2}^{(b)} \right\rbrace , \left\lbrace m_{3}^{(a)}, m_{3}^{(b)} \right\rbrace \right\rbrace, \\ 
\nonumber \Omega_1^{(9)} &= \left\lbrace \left\lbrace m_{1}^{(a)}, m_{1}^{(b)} \right\rbrace , \left\lbrace m_{2}^{(a)}, m_{2}^{(b)} \right\rbrace , \left\lbrace m_{3}^{(a)}, m_{3}^{(b)} \right\rbrace \right\rbrace.  
\end{align} 
By swapping the appropriate indices, one can similarly obtain $9$ possible choices of $\Omega_2$ and $\Omega_3$, at BSs $2$ and $3$, respectively. By evaluating only specific choices of $\Omega_1$ and their counterparts at receivers $2$ and $3$, it is readily verified that regions TIN, SD and SND are special cases of RS. For instance, if $\Omega_1^{(1)}$ is picked by the decoder at receiver $1$, while its counterparts are picked at receivers $2$ and $3$, then the RS scheme will be equivalent to TIN. However, one should note that neither TIN nor SD/SND can provide the decoding flexibilities enabled by $\Omega_1^{(2)}, \Omega_1^{(3)}, \Omega_1^{(6)}, \Omega_1^{(7)}, \Omega_1^{(8)}$. Therefore, by taking the union over all possible choices of $\mu$, the proposed RS scheme enlarges the region achieved by SND, and can thus outperform TIN, SD and SND.
\begin{remark}\label{rs:snd}
Note that if we choose $R_l^{(a)} = 0, l=1, ..., L,$ in the code construction, then $m_l^{(a)}=1$, and the codewords are $\pmb{s}_l ( 1, m_{l}^{(b)}) = \pmb{s}_l^{(b)} (m_{l}^{(b)}) + \pmb{s}_l^{(a)} ( 1 ,m_{l}^{(b)}) := \pmb{s}_l (m_l^{(b)}), l=1, ..., L,$ and these are all i.i.d Gaussian. Since in cell $l$, the messages $m_j^{(b)}, j \neq l$ (i.e., messages from the other cells) are decoded non-uniquely, and also $m_l^{(a)}$ has only one possible value (and is thus trivial to decode), then for all $\mu \in [0, 1]$, $\mathcal{S}_l^{\rm RS}$ in \eqref{eq:sl} effectively becomes the set of feasible message combinations of the form $\lbrace \pmb{m}_{\Omega}^{(b)} \rbrace_{\lbrace l \rbrace \subseteq \Omega \subseteq \mathcal{L} }$, where $\pmb{m}_{\Omega}^{(b)}$ is the vector of size $\vert \Omega \vert$ with entries $m_j^{(b)}, j \in \Omega$. Also, as pointed out in \cite[Section~2]{bandemer2015optimal}, the region given by the resulting set of constraints is equivalent to the SND region, i.e., $\left[ R_1, ..., R_L  \right] \in \mathscr{R}^{\rm SND} \Longleftrightarrow \left[0, R_1, ..., 0, R_L  \right] \in \mathscr{R}^{\rm RS}, \; \forall \mu \in [0, 1]$.
\end{remark}
\section{Maximum Symmetric Rate Allocation}
We study the maximum symmetric rate allocation problem, and compare the performance of RS with TIN, SD and SND based on the maximum symmetric rate they can offer. Computing the maximum symmetric rate over the SD region has been discussed in \cite{8913624} using the properties of convex polytopes. In the case of SND, however, one should first note that the achievable region at each receiver is the union of a finite number of MAC regions (cf. \eqref{eq:24}). Therefore, to find the maximum symmetric rate of SND at each receiver, one can calculate the maximum symmetric rate over each of these MAC regions, and then pick the largest of these quantities. 

In the case of RS, we first fix $\mu$ and solve the following problem
\begin{align}\label{chap4:eq31}
[\mathcal{P}1] \quad \quad \max& \quad \min_l \quad R_l^{ (a)} + R_l^{ (b)} \\
\quad\textrm{subject to}& \quad  \left[ R_1^{ (a)}, R_1^{ (b)}, ..., R_L^{ (a)}, R_L^{ (b)} \right] \in \mathscr{R}^{\rm RS}, \label{chap4:eq32} \\
&\quad \;\;\;  R_l^{ (a)}, R_l^{ (b)} \geq 0, \quad \forall l \in \mathcal{L}.
\end{align} 
Note that the region $\mathcal{R}_{\textrm{MAC}(\Omega_l, l ) }^{\rm RS}$ in \eqref{ch4:eq:24} is in the form of a convex polytope and the intersection of a finite number of these convex polytopes yields another convex polytope. Therefore, by distributing the intersection in \eqref{ch4:eq:23} over the union in \eqref{ch4:eq:24} (using the distributive law) the network-wide region $\mathscr{R}^{\rm RS}$ can be re-written as the union of a finite number of convex polytopes, i.e., $\mathscr{R}^{\rm RS} = \bigcup_n \tilde{\mathscr{R}}_n^{\rm RS}, \; n \in \mathcal{I}^{\rm RS}:=\lbrace 1, 2, ..., N^{\rm RS}  \rbrace $, where $N^{\rm RS}$ is the total number of these convex polytopes, and each $n$ corresponds to a unique choice of $\left( \Omega_1, \Omega_2, ..., \Omega_L \right) \in \mathcal{S}_1^{\rm RS} \times ... \times \mathcal{S}_L^{\rm RS}$. Note that solving the maximum symmetric rate problem over one of these convex polytopes can be formulated as an LP. Specifically, we first define the two $(2L + 1) \times 1$ (where $L \geq 2$) column vectors $\pmb{x} := [ R_1^{ (a)}, R_1^{ (b)}, ..., R_L^{ (a)}, R_L^{ (b)}, t ]^T$ and $\pmb{c} := \left[ 0, ..., 0, 1 \right]^T$. Then, the convex polytope $\tilde{\mathscr{R}}_n^{\rm RS}$ can be written in matrix form as $\tilde{\pmb{A}_n} \pmb{x} \leq \tilde{\pmb{b}}_n(\mu)$, where $\leq$ denotes element-wise inequality between two vectors, and matrix $\tilde{\pmb{A}_n}$ and vector $\tilde{\pmb{b}}_n(\mu)$ are constructed as follows. As explained below \eqref{ch4:eq:24}, with each choice of $\Omega_l$, a modified MAC region $\mathcal{R}_{\textrm{MAC}(\Omega_l, l ) }^{\rm RS}$ is obtained, which corresponds to a system of linear inequalities, i.e., $\pmb{A}_{\Omega_l} \pmb{x} \leq \pmb{{b}}_{\Omega_l} (\mu)$. For instance, in the case of $L=2$, $|{\cal S}_1^{\rm RS}| = 3$ and therefore there are 3 possible modified MAC regions at receiver 1, which are given by \eqref{eq:app:rs:5}-\eqref{app:rs:5:2}, \eqref{eq:app:rs:9}-\eqref{app:rs:8:5} and \eqref{eq:app:rs:12}-\eqref{app:rs:12:11} in Appendix~A-3. $\pmb{\tilde{A}}_n$ and $\pmb{\tilde{b}}_n(\mu)$ are then obtained by stacking
$\pmb{A}_{\Omega_l}, l = 1, \ldots, L,$ and $\pmb{{b}}_{\Omega_l}(\mu), l = 1, \ldots, L$, respectively, i.e.,
\begin{align}
\pmb{\tilde{A}}_n
= \left[
 \pmb{A}_{\Omega_1}^T , \hdots , \pmb{A}_{\Omega_L}^T 
\right]^T ,\;
\pmb{\tilde{b}}_{n}(\mu)
= \left[
 \pmb{{b}}_{\Omega_1}^T (\mu) , \hdots , \pmb{{b}}_{\Omega_L}^T (\mu) 
\right]^T. 
\end{align}
In addition, we have the inequality constraints $- R_{l}^{ (a)} - R_{l}^{ (b)} + t \leq 0, l=1, ..., L$, that can be written in matrix form as $\pmb{D} \pmb{x} \leq \pmb{0}$, where $\pmb{D}$ is of size $L \times (2L+1)$ with the last entry of each row always being $1$, i.e.,
\begin{equation}\label{eq:1}
\pmb{D} = 
\begin{pmatrix}
-1 & -1 & 0 & 0 & 0 & \cdots & 0 & 1 \\
0 & 0 & -1 & -1 & 0 & \cdots & 0 & 1 \\
\vdots  & & &   \ddots & \ddots & &  & \vdots \\
0 & \cdots & & & 0 & -1 & -1 & 1 
\end{pmatrix}.
\end{equation}
Then, the equivalent optimization problem is obtained as follows
\begin{align}\label{chap4:eq33}
[\mathcal{P}1^{\prime}] \quad\quad &\max_{n \in \mathcal{I}^{\rm RS}} \quad\quad\quad\quad \max_{ \pmb{x}} \quad  \pmb{c}^T \pmb{x} \\
&\textrm{subject to} \quad \quad \; \;\; \tilde{\pmb{A}_n} \pmb{x} \leq \tilde{\pmb{b}}_n(\mu) \label{chap4:eq33:1} \\
&\hspace{24mm} \pmb{D} \pmb{x} \leq \pmb{0}, \quad \pmb{x} \geq \pmb{0},  
\end{align}
where the inner problem for a fixed $n$ is an LP, and the outer maximization finds the index $n$ of the polytope $\tilde{\mathscr{R}}_n^{\rm RS}$ that gives rise to the best symmetric rate for a fixed $\mu$. Denote the optimal value of $[\mathcal{P}1^{\prime}]$ by $t^{\ast} (\mu)$. Noting that the overall region is obtained by taking the union over the combination of all possible power splitting strategies, the optimal solution to the symmetric rate problem is found as 
\begin{align}\label{chap4:eq35}
 &\max_{0 \leq \mu \leq 1} \quad  t^{\ast}(\mu). 
\end{align}   
For networks with large number of cells (e.g., $L=7$), searching over all sub-regions (i.e., convex polytopes $\tilde{\mathscr{R}}_n^{\rm RS}$) to find the best symmetric rate may not be computationally feasible. In the following, we introduce a subset of these sub-regions that provides an achievable lower bound to the performance of RS, and numerically show that this subset still offers a significant gain over TIN, SD and SND.

First, define a subset of $\mathcal{S}_l^{\rm RS}$ as follows 
\begin{align}
\nonumber \mathcal{S}_l^{\rm Sub} &:= \left\lbrace  \left\lbrace m_l^{\rm (a)}, m_l^{\rm (b)}  \right\rbrace  \right\rbrace \\ 
\nonumber &\hspace{6mm}\times \Big\lbrace \left\lbrace  m_1^{\rm (b)} \right\rbrace, ..., \left\lbrace m_{l-1}^{\rm (b)} \right\rbrace, \left\lbrace m_{l+1}^{\rm (b)} \right\rbrace, ..., \left\lbrace m_L^{\rm (b)} \right\rbrace, \\ 
&\hspace{13mm}\left\lbrace m_1^{\rm (b)}, ..., m_{l-1}^{\rm (b)}, m_{l+1}^{\rm (b)}, ..., m_L^{\rm (b)} \right\rbrace  \Big\rbrace, \label{eq:subset}
\end{align} 
which gives rise to $\mathscr{R}^{\rm Sub} \subset \mathscr{R}^{\rm RS}$ as follows
\begin{equation}\label{ch4:eq:sub1}
\mathscr{R}^{\rm Sub} = \bigcap_{l=1}^L \mathcal{R}_{l}^{\rm Sub},
\end{equation}
where 
\begin{equation}\label{ch4:eq:sub2}
\mathcal{R}_{l}^{\rm Sub} = \bigcup_{ \Omega_l \in \mathcal{S}_l^{\rm Sub}  } \mathcal{R}_{\textrm{MAC}(\Omega_l, l ) }^{\rm RS}.
\end{equation}
This sub-region $\mathscr{R}^{\rm Sub}$ of $\mathscr{R}^{\rm RS}$ can be represented as the union of a finite number of convex polytopes, i.e., $\mathscr{R}^{\rm Sub} = \bigcup_{n} \mathscr{\tilde{R}}_n^{\rm RS}, \; n \in \mathcal{I}^{\rm Sub}$, where $\mathcal{I}^{\rm Sub} \subset \mathcal{I}^{\rm RS}$ since $\mathcal{S}_l^{\rm Sub} \subset \mathcal{S}_l^{\rm RS}$. Further define the function $g (\mathcal{I}) := \max_{0 \leq \mu \leq 1} \; f \left( \mu, \mathcal{I} \right) $, where $f \left( \mu, \mathcal{I} \right)$ is given by
\begin{align}\label{ch4:eq42:11}
f \left( \mu, \mathcal{I} \right) := \quad &\max_{n \in \mathcal{I}} \quad\quad\quad\quad \max_{ \pmb{x}} \quad  \pmb{c}^T \pmb{x} \\
&\textrm{subject to} \quad \quad \;  \tilde{\pmb{A}_n} \pmb{x} \leq \tilde{\pmb{b}}_n(\mu) \label{chap4:eq42:1} \\
&\hspace{22mm} \pmb{D} \pmb{x} \leq \pmb{0}, \quad \pmb{x} \geq \pmb{0}. \label{ch4:eq42:20} 
\end{align}
For instance, the maximum symmetric rate problem over the RS region in \eqref{chap4:eq35} can be written as $g (\mathcal{I^{\rm RS}}) = \max_{0 \leq \mu \leq 1} \; f \left( \mu, \mathcal{I^{\rm RS}} \right) $. Since $\mathcal{I}^{\rm Sub} \subset \mathcal{I}^{\rm RS}$, we have $f \left( \mu, \mathcal{I^{\rm RS}} \right) \geq  f \left( \mu, \mathcal{I^{\rm Sub}} \right)$, and $f \left( \mu, \mathcal{I^{\rm Sub}} \right)$ is achievable by RS for all $\mu \in [0,\; 1]$. In addition, from Remark~\ref{rs:snd}, for all $\mu \in [0,\; 1]$, the SND region is a special case of the RS region. It then follows that $f \left( \mu, \mathcal{I^{\rm RS}} \right) \geq t^{\rm SND, \ast}$, where $t^{\rm SND, \ast}$ is the maximum symmetric rate over the SND region (also achievable by RS). Consequently, we have
\begin{align}
g (\mathcal{I^{\rm RS}}) \; &\geq \;  \max_{ 0 \leq \mu \leq 1} \;\; \max \left\lbrace   \; t^{\rm SND, \ast} , \;  \; f \left( \mu, \mathcal{I^{\rm Sub}} \right)   \right\rbrace \label{ch4:eq50} \\
&= \; \max \; \left\lbrace t^{\rm SND, \ast}, \; g (\mathcal{I^{\rm Sub}})  \right\rbrace. \label{ch4:eq51}  
\end{align}Therefore, by computing $\max \left\lbrace t^{\rm SND, \ast} , g \left( \mathcal{I}^{\textrm{Sub}} \right) \right\rbrace$, an achievable lower bound to the true performance of the RS scheme is obtained.
\begin{remark}
For comparison, consider the case of a more complicated power splitting scheme as briefly mentioned in Section~\ref{sec:beyond:two}, where each of the $L$ users is allowed to have a different splitting coefficient, $\mu_l$, so each user can better fine tune its allotted inner and outer layer power for rate optimization. While this scheme leads to additional degrees of freedom  and can thereby potentially provide better performance, the search region to compute $g (\mathcal{I^{\rm Sub}}) $ in \eqref{ch4:eq51} is now $(\mu_1, \ldots, \mu_l) \in [0,1]^L$ instead of $\mu \in [0,1]$. If each interval $[0,1]$ is divided into $N$ steps, the search space for $(\mu_1, \ldots, \mu_l)$ is then considerably larger than that of $\mu$, even for moderate $N$ and $L$.
\end{remark}

We will see in the next section that searching over $\mu \in [0,\; 1]$ in \eqref{ch4:eq51} is not necessarily needed. Specifically, one can skip numerically optimizing $\mu$ and instead use a pre-computed average value of the optimized splitting coefficients obtained in different random realizations with negligible performance loss; hence reducing the optimization search space. 

We now briefly discuss the computational complexity of finding the maximum symmetric rate for SD, SND and RS by counting the number of LPs required to be solved in each case\footnote{In the case of TIN, finding the maximum symmetric rate is much simpler as 
this is given by the minimum of $L$ single-user rate constraints.}. For SD, we have one MAC region achieved at each receiver of the $L$-user IC. Thus, the maximum symmetric rate can be found by solving $L$ LPs, one at each receiver. In the case of SND, from \eqref{eq:24}, it can be seen that the achievable region at each receiver is the union of $2^{L-1}$ MAC regions, and over each of them the maximum symmetric rate problem is a standard LP. Therefore, across all $L$ receivers, we need to solve $L \times 2^{L-1}$ LPs. For RS, in addition to finding the maximum symmetric rate of SND, $t^{\rm SND, \ast}$, one also needs to compute $g({\cal I}^{\rm Sub}) = \max_\mu f(\mu, \cal{I}^{\rm Sub})$  in \eqref{ch4:eq51}. Due to \eqref{ch4:eq42:11}-\eqref{ch4:eq42:20}, for a given $\mu$, solving for $f(\mu, {\cal I}^{\rm Sub})$ requires solving $|{\cal I}^{\rm Sub}|=L^L$ LPs, one for each $n \in {\cal I}^{\rm Sub}$. If the interval $[0,1]$ is divided into $N$ steps to perform the line search for $\mu$, then computing $g(\cal{I}^{\rm Sub})$ requires solving $N\times L^L$ LPs. In the case that a pre-computed average value of $\mu$ is used, only $L^L$ LPs are solved. We also emphasize the fact that since these sub-regions and their corresponding LPs do not depend on each other, the LP problems could be solved in parallel.
\section{Simulation Results}\label{sec:simulation}
\sloppy To compare the performance of the different schemes, TIN, SD, SND and RS, with maximum symmetric SE (in units of bits/sec/Hz) allocation, we simulate the downlink of a seven-cell massive MIMO system (i.e., $L=7$). In particular, we consider seven hexagons with wrap around topology where the cell radius is $400$ m and one BS is located at the center of each cell. Also, $K=15$ users are uniformly distributed at random within the area of each cell, but at least $35$ m away from the cell center. Moreover, we take the average of the maximum symmetric SEs over $150$ random realizations of user locations. The BS transmit power is taken to be $40$ W (46 dBm), and a 3-dimensional distance-based path-loss model adopted from \cite{pathloss} is used to model large-scale fading coefficient, $\beta_{jkl}$:
\begin{align}
\nonumber \left[ \beta_{jkl} \right]_{\rm dB} &= -13.54 -39.08 \log_{10} \left( d_{jkl}^{3D}  \right) \\
 &\hspace{4mm}- 20 \log_{10} \left( f_c \right) + 0.6 \left(  h_{UT} - 1.5 \right),\label{pathloss}
\end{align}
where $d_{jkl}^{3D}$ is the 3D distance (in meters) from user $k$ in cell $l$ to BS $j$, the carrier frequency is $f_c = 3.5$ GHz, $h_{UT}$ is the user height which is taken to be $1.5$ m, while the BS height is $25$ m. Also, the noise variance is assumed to be $-101$ dBm. Note that while the effects of shadowing are omitted in \eqref{pathloss}, we will investigate its impact on system performance separately at the end of this section by adding an extra term associated with log-normal shadowing to \eqref{pathloss}. Below, we illustrate the two cases of a spatially correlated Rayleigh fading channel and an uncorrelated Rayleigh fading channel separately.
\subsection{Spatially Correlated Channel}
We now study the downlink performance of ZF when a spatially correlated channel model is used. We adopt the exponential correlation model of \cite{loyka2001channel}, i.e., $\pmb{R}_{jkl}$ is a Hermitian Toeplitz matrix with the first row given by $[1, r_{jkl}^{\ast}, \cdots , ( r_{jkl}^{\ast} )^{M-1} ]$, which is widely used in the literature \cite{bjornson2018massive, van2018large, bjornson2017random}. In particular, in this model $r_{jkl} = \kappa e^{j \phi_{jkl}}$ is the correlation coefficient, $\kappa \in [0,\; 1]$ is the correlation magnitude and $\phi_{jkl}$ is the user angle to the antenna array boresight. Unless otherwise specified, we assume $\kappa =0.4$, i.e., moderate spatial correlation. 

To compute the maximum symmetric SE of the RS scheme, two different approaches are utilized. Specifically, in the first approach, an achievable lower bound to the maximum symmetric SE is found by solving \eqref{ch4:eq51} and numerically searching over $0 \leq \mu \leq 1$ (with a step size of $0.02$) to find the optimum value of the power splitting coefficient. For each value of $M$, the average of the optimum choices of $\mu$ over $150$ realizations is also calculated and stored. In the second approach, rather than numerically optimizing $\mu$ in \eqref{ch4:eq51}, an achievable lower bound is calculated based on \eqref{ch4:eq51} but using the pre-computed average value of $\mu$, which is validated on $150$ new random realizations of user locations. As such, the computational cost of numerically optimizing $\mu$ in the first approach is now reduced in the second approach. 
 
Fig.~\ref{seven:zf} shows the performance of the different schemes with ZF, where the achieved SEs for RS are obtained from the two approaches explained above. Interestingly, it is revealed from these figures that for each $M$ it is sufficient to use only the pre-computed average values of $\mu$ as in the second approach. In other words, calculating the SEs using the second approach yields almost the same performance as that obtained from the first approach, showing the advantage of using pre-computed average values of $\mu$ in practical implementations of RS.
\begin{figure*}[t!] 
	\captionsetup[subfloat]{captionskip=3mm}
	\centering\hspace{-6mm}   
	\subfloat[\footnotesize{ZF}]{\includegraphics[scale=0.64]{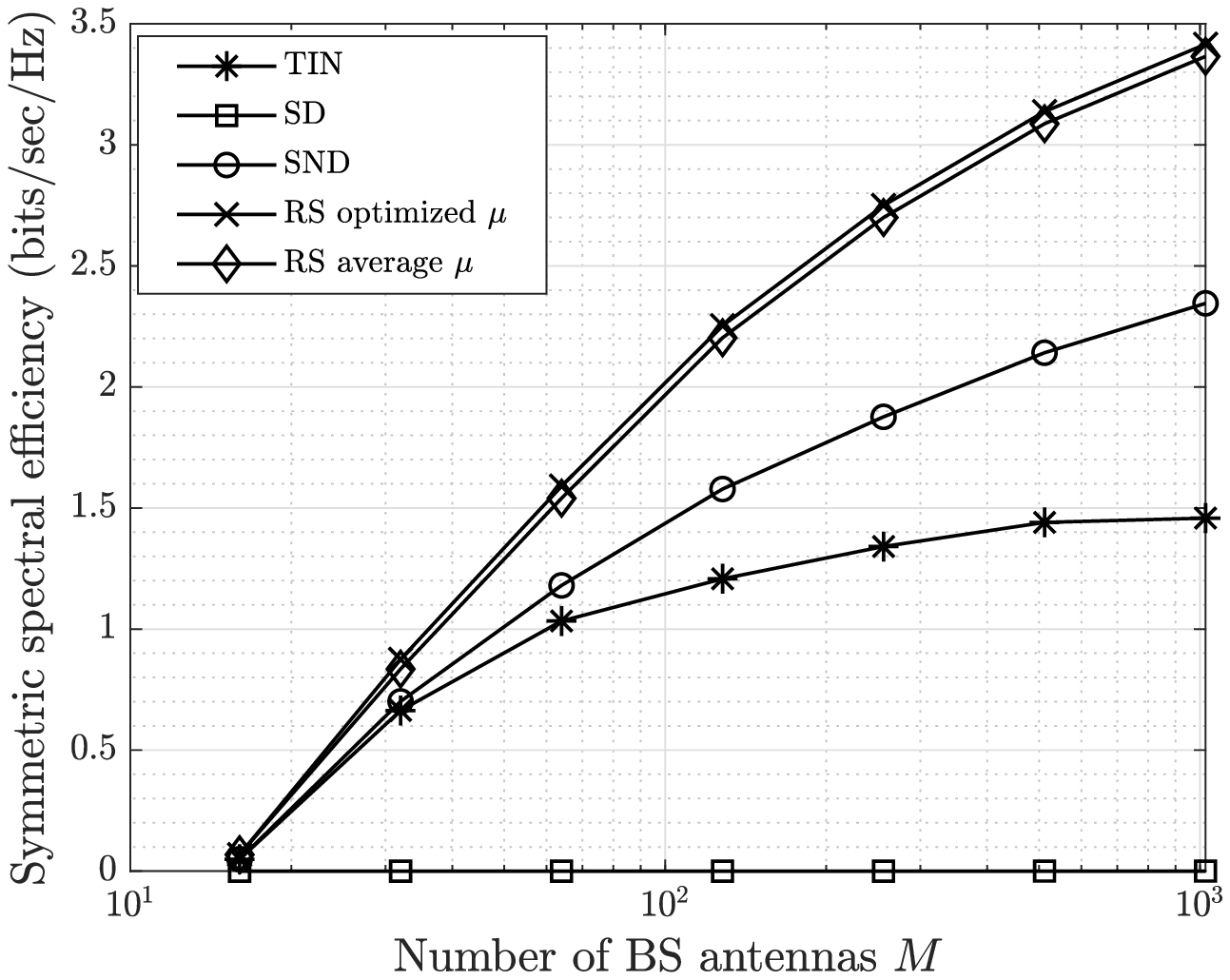}\label{seven:zf}}\hspace{-5mm}
	\subfloat[\footnotesize{RZF vs. ZF}]{\includegraphics[scale=0.64]{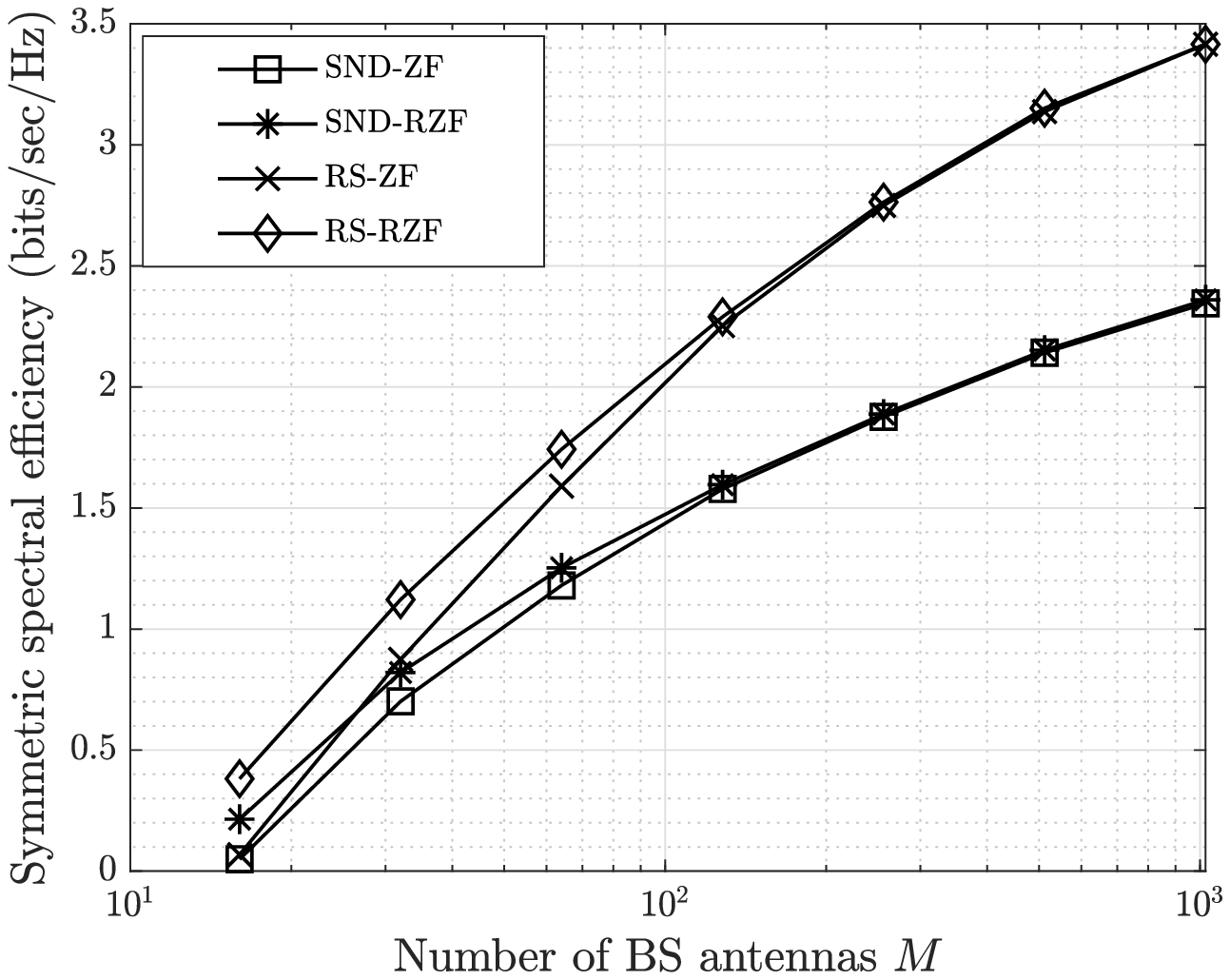}\label{fig:12:1}} 
	\caption{Performance of various schemes using maximum symmetric SE, when a spatially correlated Rayleigh fading channel model is used: (a) TIN, SD, SND and RS, with ZF precoding, (b) SND and RS with either ZF or RZF. \label{performance:correlated}}
\end{figure*}

Fig.~\ref{seven:zf} also confirms that while both cases of SND and RS outperform TIN and their performance improves by increasing $M$, due to the additional flexibilities enabled by partial decoding, RS achieves significantly larger SEs than SND. In particular, while the gain provided by SND over TIN is about $30\%$ and $40\%$ for $M=128$ and $M=256$, respectively, and increases to about $60\%$ when $M=1024$, this gain for RS is at least\footnote{Recall that \eqref{ch4:eq51} provides an achievable lower bound on the performance of RS, and therefore the actual gain over TIN may be larger.} $86\%$ and $105\%$ for $M=128$ and $M=256$, respectively, and increases to at least $134\%$ when $M=1024$. Lastly, this figure confirms that SD performs poorly compared to other schemes, as it tries to \textit{uniquely} decode all pilot contamination interference terms regardless of their strength. This is as opposed to SND and RS, where pilot contamination interference terms are decoded non-uniquely either in their \textit{entirety} or \textit{partially}. In fact, as we will see in the sequel, SD can outperform TIN only when $M$ is extremely large and thus beyond practical limits.

Next, we investigate the performance of regularized zero-forcing (RZF), where the precoding matrix at BS $j$ is given by
\begin{equation}\label{rzf}
\pmb{W}_{j}^{\rm rzf} = \pmb{\hat{G}}_{jj} \left( \pmb{\hat{G}}_{jj}^{\dagger} \pmb{\hat{G}}_{jj} + \delta \pmb{I}_K \right)^{-1}, \quad \forall j,
\end{equation}
where $\delta$ is a regularization factor. Note that the choice of $\delta$ is arbitrary and could be further optimized (see for example \cite[Theorem~6]{jose2011pilot} and \cite{nguyen2019multi}). The two choices of $\delta=K / \rho_{\rm dl}$ (suggested by  \cite{peel2005vector}) and $\delta= M / \rho_{\rm dl}$ (suggested by \cite{hoydis2013massive}) were explored by simulation. The former provided better performance for the setup and system parameters considered. Therefore, in this paper we take $\delta = K / \rho_{\rm dl}$. It can be verified that, when $M$ is large, the diagonal entries of $\pmb{\hat{G}}_{jj}^{\dagger} \pmb{\hat{G}}_{jj}$ increase with $M$ and therefore the approximation $(  \pmb{\hat{G}}_{jj}^{\dagger} \pmb{\hat{G}}_{jj} + \delta \pmb{I}_K )^{-1} \approx ( \pmb{\hat{G}}_{jj}^{\dagger} \pmb{\hat{G}}_{jj} )^{-1}$ can be used. Hence, for large $M$, one expects the performance of RZF to resemble that of ZF. On the other hand, when $M$ is small, with a proper choice of the regularization factor, RZF can outperform ZF \cite{peel2005vector}. These results are confirmed in Fig.~\ref{fig:12:1}. In particular, this figure shows the performance of SND and RS when either ZF or RZF is applied at the BSs. It can be observed that when $M$ is moderately small (i.e., $M \leq 64$), there is a visible gain offered by RZF, while for large $M$ the performance of RZF converges to that of ZF. Moreover, RS achieves significantly higher SEs compared to SND, which is similar to the observations with ZF.     

Fig.~\ref{impact-kappa-user} shows the impact of changing the antenna correlation magnitude $\kappa$ and the number of users $K$ on performance of the different schemes, where $M=256$ and ZF precoding is used. The SEs achieved by SD are not shown here as it performs poorly for practical values of $M$. It is evident from Fig.~\ref{rate-magnitude} that increasing the correlation magnitude from $0$ (i.e., uncorrelated fading) to $0.8$ (i.e., strong spatial correlation) results in improving the performance of all schemes; hence, reducing the achieved performance gaps. Specifically, it can be observed that while SND offers a gain of $61\%$ over TIN in the uncorrelated regime ($\kappa =0$), RS again provides superior performance with a gain of at least $140\%$ in this regime. On the other hand, in the strong spatial correlation regime ($\kappa = 0.8$), while the gain of SND over TIN reduces to $11\%$, RS now provides a gain of at least $35\%$ over TIN. 
\begin{figure*}[t!] 
	\captionsetup[subfloat]{captionskip=3mm}
	\centering\hspace{-6mm}    
	\subfloat[\footnotesize{Impact of varying $\kappa$.}]{\includegraphics[scale=0.64]{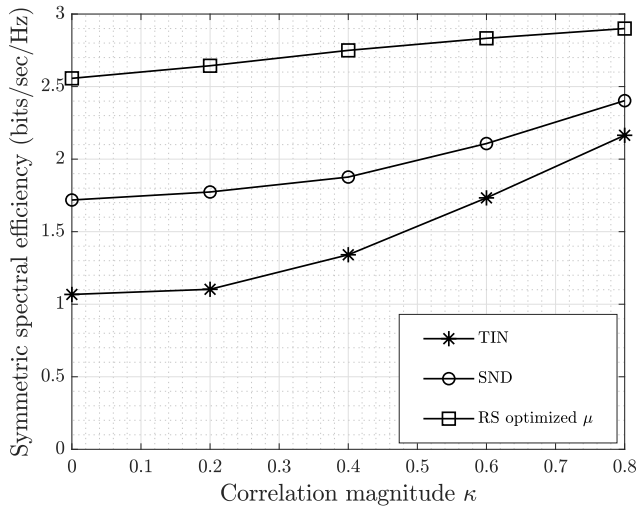}\label{rate-magnitude}}\hspace{-5mm} %
	\subfloat[\footnotesize{Impact of varying $K$.}]{\includegraphics[scale=0.64]{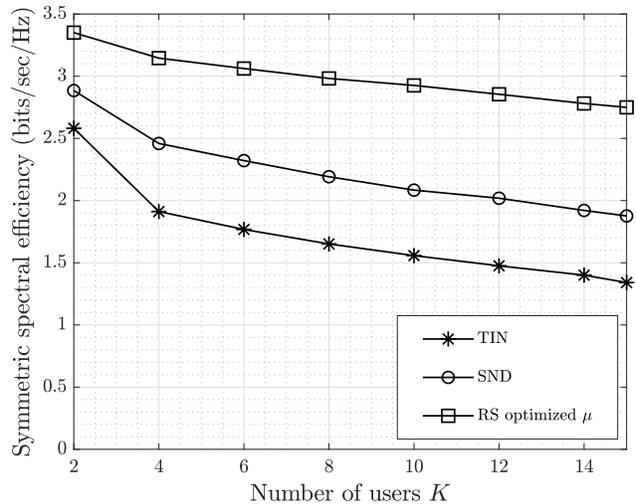}\label{rate-user}}
    \caption{Impact of increasing the correlation magnitude and the number of users on performance of TIN, SND and RS, with ZF precoding, when $M=256$: (a) $\kappa$ is varying while $K=15$, (b) $K$ is varying while $\kappa=0.4$. \label{impact-kappa-user}}
\end{figure*} 
The performance improvement seen by increasing the spatial correlation of the channel is in agreement with the results reported in \cite{bjornson2017book}. Particularly, it is known that spatial correlation can improve the quality of MMSE channel estimates; thus, resulting in reduced pilot contamination effects in massive MIMO systems. Therefore, it is expected that pilot contamination interference causes its most adverse impact when channel correlation is zero, and increasing spatial correlation alleviates this problem gradually.

In addition, Fig.~\ref{rate-user} reveals that increasing $K$ results in degrading the performance of all schemes. This is as expected, since serving a larger number of users leads to smaller symmetric SEs. Nonetheless, as $K$ increases, the achieved performance gains over TIN improve. In particular, while SND provides a gain of $11\%$ and $40\%$ for $K=2$ and $K=15$, respectively, the gain of RS over TIN is significantly better and at least $31\%$ for $K=2$, and increasing to at least $105\%$ when $K=15$.

Lastly, we study the impact of shadow fading on performance of the proposed schemes. In particular, we assume that a term associated with shadow fading is now added to the large-scale fading model of \eqref{pathloss} with a standard deviation of $\sigma_{\rm shadow}$ in dB. Fig.~\ref{rate-shadow} shows the achieved symmetric SEs of TIN, SND and RS, where the standard deviation of the shadow fading, $\sigma_{\rm shadow}$, varies from $0$ dB to $5$ dB. The parameters for this figure are the same as those in Fig.~\ref{rate-magnitude}, except that the correlation magnitude is now fixed at $\kappa = 0.4$. It can be observed that, as expected, by increasing shadow fading the SEs achieved by all schemes reduce. Nevertheless, as $\sigma_{\rm shadow}$ becomes larger the gains provided by SND and RS over TIN increase, with RS achieving superior performance compared to TIN and SND in all cases. In particular, when there is no shadowing in the path-loss model of \eqref{pathloss}, the gain provided by SND over TIN is $40\%$ and improves to more than a factor of $3$ when shadowing increases to $\sigma_{\rm shadow} = 3$ dB, whereas in the case of RS this gain is at least $105\%$ without shadowing effects and improves to more than a factor of $5$ when $\sigma_{\rm shadow} = 3$ dB. Furthermore, these gains continue to grow for larger values of $\sigma_{\rm shadow}$.
\begin{figure}[t!]
\centering
\hspace{-6mm}\includegraphics[scale=0.64]{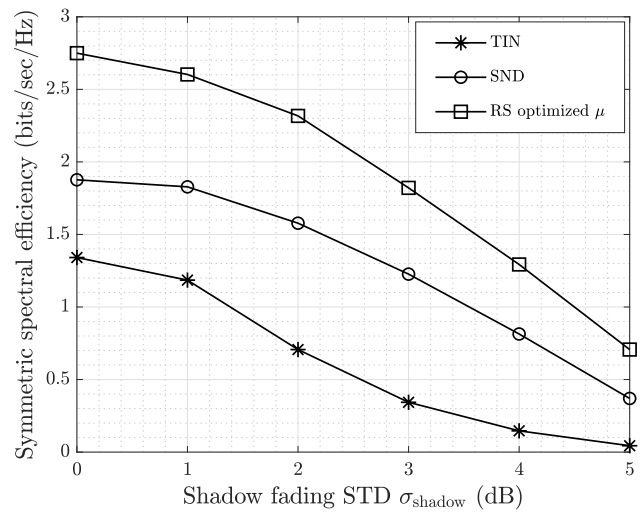}
\caption{Impact of increasing the standard deviation of shadow fading $\sigma_{\rm shadow}$ on performance of TIN, SND and RS, with ZF precoding, when $M=256$, $K=15$ and $\kappa=0.4$.\label{rate-shadow}}
\end{figure}
\subsection{Uncorrelated Channel}
We now consider the special case of uncorrelated Rayleigh fading, i.e., $\pmb{R}_{jkl} = \beta_{jkl} \pmb{I}_M$, with ZF precoding where the simulation parameters are the same as those in Fig.~\ref{performance:correlated}. To evaluate the performance, the average of the maximum symmetric SEs is calculated over $200$ random realizations of user locations. Also, using the closed-form expression of the rate lower bound in \eqref{thm:down:zf:1} for an uncorrelated channel, we are able to compute the performance for a significantly wider range of $M$, thus providing insights into the asymptotic performance limits of the different schemes. 

Fig.~\ref{small-m} shows these results for a range of moderately large $M$, while Fig.~\ref{large-m} shows the same for a range of extremely large $M$. While the latter covers a range of $M$ that is beyond practical values, the results of Fig.~\ref{large-m} can be used to confirm asymptotic performance limits as $M \rightarrow \infty$. Similar to the case of a correlated channel, it is evident that while the performance of all interference decoding schemes improves with increasing $M$, RS achieves significantly larger SEs compared to all other schemes. Furthermore, it can be seen that the symmetric SEs obtained using the optimized values of splitting coefficients for RS are almost the same as those obtained using the pre-computed average values; thus, reducing the optimization search space. Fig.~\ref{small-m} shows that SND provides a gain of about $45\%$ and $61\%$ over TIN when $M=128$ and $M=256$, respectively, and this gain reaches about $94\%$ when $M=1024$. On the other hand, due to the advantages offered by partial decoding, the gain provided by RS over TIN increases to at least $109\%$ and $140\%$ for $M=128$ and $M=256$, respectively, and improves to at least $191\%$ when $M=1024$, which is again much larger than SND. 
\begin{figure*}[t!] 
	\captionsetup[subfloat]{captionskip=3mm}
	\centering\hspace{-6mm}    
	\subfloat[\footnotesize{Moderately large $M$.}]{\includegraphics[scale=0.64]{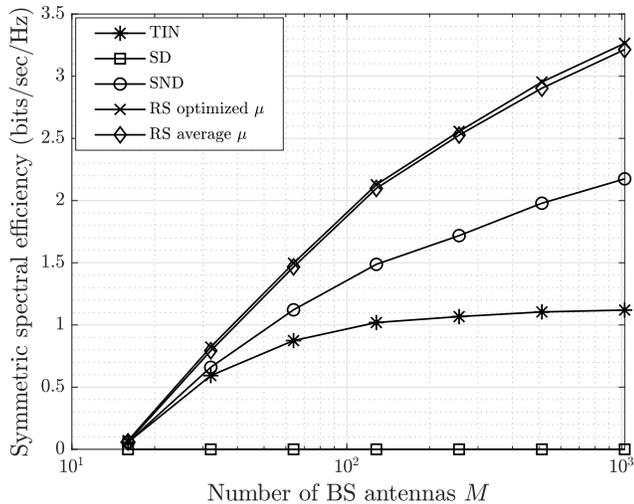}\label{small-m}}\hspace{-5mm} %
	\subfloat[\footnotesize{Extremely large $M$.}]{\includegraphics[scale=0.64]{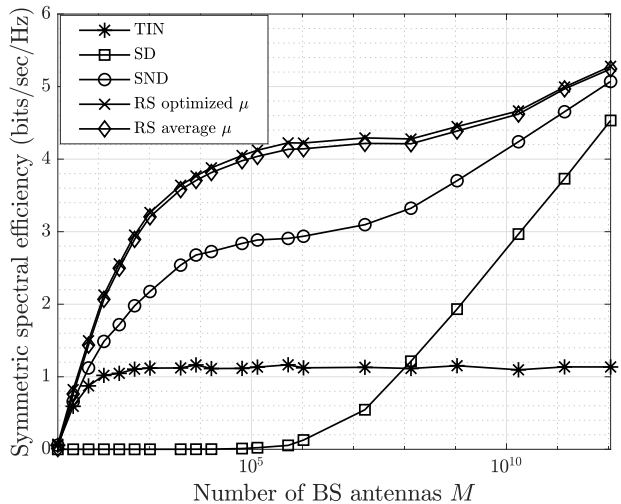}\label{large-m}}
	\caption{Performance of TIN, SD, SND and RS, using maximum symmetric SE, when the special case of an uncorrelated Rayleigh fading channel model is used: (a) Achieved symmetric SEs over the range of moderately large $M$, (b) Achieved symmetric SEs over the range of extremely large $M$. \label{uncorrelated}}
\end{figure*} 

It is also confirmed via Fig.~\ref{large-m} that, when $M$ grows unbounded, the performance of TIN saturates to a constant value, while the performance of SD, SND and RS continue to grow logarithmically with $M$, as expected. Moreover, it can be observed that when $M$ is extremely large (i.e., approximately $M > 10^8$) SD starts to outperform TIN, as it now becomes beneficial to uniquely decode pilot contamination interference. 

Interestingly, in Fig.~\ref{large-m}, one can notice that as $M$ increases, the performance gap between RS and SND gradually increases up to a point (i.e., approximately $ M \approx 10^5$), beyond which the gap to SND gradually diminishes, and they eventually converge. This means that as $M$ grows and thus the strength of the pilot contamination interfering signals increases, the power allocated to the inner layer (determined by $1 - \mu$) can be increased. Hence, as $M$ grows unbounded, one can non-uniquely decode the entire part of the interference terms under RS (i.e., $\mu \rightarrow 0$); thus achieving a performance close to that of SND.
\section{Conclusion}
In this paper, the problem of decoding pilot contamination interference was studied for the downlink of a multi-cell massive MIMO system. Using a worst-case uncorrelated noise technique, a general achievable rate lower bound was obtained, which was further specialized to ZF precoding for an uncorrelated channel. Moreover, a novel partial interference decoding (i.e., RS) scheme based on message splitting and non-unique decoding was proposed and an achievable rate region was established for this scheme. 
We show that finding the maximum symmetric SE over this region can be found by solving multiple LP problems.
To illustrate the performance, an achievable sub-region of the RS scheme was introduced that provides a lower bound to the performance of RS, yet achieving significantly larger maximum symmetric SEs compared to TIN, SD and SND for a practical number of antennas. Furthermore, the impact of increasing the correlation of the channel across antennas, the number of users and the degree of shadow fading was numerically investigated. In all scenarios, it was observed that RS maintains superior gain over TIN, SD and SND. It was also shown that one can replace the numerically-optimized value of the power splitting coefficient with its pre-computed average value, which gives rise to a negligible performance loss; thereby reducing the optimization search space.

One possible future direction 
may be to consider improving the lower bound in \eqref{ch4:eq51} by, for instance, enlarging the subset in \eqref{eq:subset}. Therefore, an interesting question is the investigation of which subsets to add that provide the most benefit. 
Another future direction is to consider metrics other than maximum symmetric SE. For example,  maximizing the geometric mean results in a proportional fair rate allocation. In this case, the problem is no longer an LP on each convex polytope sub-region, but nevertheless remains convex on the sub-regions. The problem may also potentially require revisiting the considered subsets in \eqref{eq:subset}.  Another direction would be to explore more complicated power and rate splits, although this may be challenging due to the size of the search space. Nevertheless, structural results and insights on how this splitting should be done can be beneficial and may reduce the optimization search space. 
Finally, the practical implementation of codes for partial interference decoding (as well as their decoding algorithms) has several  challenges. Among these are the design of rate-flexible codes with multiple layers that can adapt to frequency offsets (due to transmissions from different BSs) and timing offsets (caused by different propagation delays) between layers at different users as well as the challenge of error propagation when incorrectly decoding a layer. We conclude by highlighting the promising approach of \cite{wang2020sliding} which has proposed the use of sliding-widow codes to implement HK partial interference decoding.

\appendices 
\section{} 
\textbf{\textit{A-1) Proof of Lemma~1:}} Let $P_{1,j}, P_2, P_3$ and $P_4$ denote the variances of $\sqrt{\rho_{\rm dl}/\lambda_j} \mathbb{E} [ \pmb{g}_{jil}^{\dagger} \pmb{w}_{jij} ] s_j[i], j=1, ..., L,$ the second sum, the third sum and noise of $y_{il}$ in \eqref{eq:down:1}, respectively. 
Note that these are all zero-mean and uncorrelated. Hence, we have $\textrm{var} \left[ y_{il} \right] = P_1 + P_2+ P_3 + P_4$, where $P_1 = \sum_{j=1}^L P_{1,j}$. Since $\mathbb{E} [ \left\vert s_j [k] \right\vert^2 ] =1, \forall j, k,$ we obtain $P_{1,j} =  \sqrt{\rho_{\rm dl}/\lambda_j} \vert \mathbb{E} [ \pmb{g}_{jil}^{\dagger} \pmb{w}_{jij} ] \vert^2$. In addition, one can equivalently write $y_{il} =  \sum_{j=1}^L \sqrt{{\rho_{\rm dl}}/{\lambda_j}} \sum_{k=1}^K \pmb{g}_{jil}^{\dag} \pmb{w}_{jkj} s_j [k] + z_{il} $, and therefore obtain $\textrm{var} \left[ y_{il} \right] = \sum_{j=1}^L \sum_{k=1}^K \dfrac{\rho_{\rm dl}}{\lambda_j} \mathbb{E} [  \vert \pmb{g}_{jil}^{\dagger} \pmb{w}_{jkj} \vert^2 ] + 1,$ since $\mathbb{E} [ \vert z_{il} \vert^2 ] = 1$. Thus, $P_2 + P_3 + P_4$ is found as  
\begin{align}
\nonumber P_2 + P_3 + P_4 &= \sum_{j=1}^L \sum_{k=1}^K \dfrac{\rho_{\rm dl}}{\lambda_j} \mathbb{E} \left[  \left\vert \pmb{g}_{jil}^{\dagger} \pmb{w}_{jkj} \right\vert^2 \right] \\
&\hspace{3mm}+ 1 - \sum_{j=1}^L \dfrac{\rho_{\rm dl}}{\lambda_j} \left\vert \mathbb{E} \left[ \pmb{g}_{jil}^{\dagger} \pmb{w}_{jij} \right] \right\vert^2. \label{app:a:2:4}
\end{align} 
Noting that $\lbrace s_j[i] \rbrace_{j \in \Omega}$ are the only signals that are decoded jointly, the remaining signals in the first sum in \eqref{eq:down:1} will be treated as noise in the lower bound of \eqref{eq:thm:2:1}. Hence, to compute the denominator, one should add $\sum_{j \in \Omega^c} P_{1,j}$ to \eqref{app:a:2:4}, i.e.,
\begin{align}
\nonumber P_2 + P_3 + P_4 + \sum_{j \in \Omega^c} P_{1,j} &= \sum_{j=1}^L \sum_{k=1}^K \dfrac{\rho_{\rm dl}}{\lambda_j} \mathbb{E} \left[  \left\vert \pmb{g}_{jil}^{\dagger} \pmb{w}_{jkj} \right\vert^2 \right] \\
\nonumber &\hspace{3mm}+ 1 - \sum_{j \in \Omega} \dfrac{\rho_{\rm dl}}{\lambda_j} \left\vert \mathbb{E} \left[ \pmb{g}_{jil}^{\dagger} \pmb{w}_{jij} \right] \right\vert^2  . \label{app:a:2:4:1} 
\end{align}
Lastly, for the numerator in the r.h.s of \eqref{eq:thm:2:1} we have 
\begin{equation}\label{app:a:2:5}
\textrm{var} \left[ \sum_{j \in \omega}  \sqrt{\dfrac{\rho_{\rm dl}}{\lambda_j}} \mathbb{E} \left[ \pmb{g}_{jil}^{\dagger} \pmb{w}_{jij} \right] s_j[i] \right] = \sum_{j \in \omega} \dfrac{\rho_{\rm dl}}{\lambda_j} \left\vert \mathbb{E} \left[ \pmb{g}_{jil}^{\dagger} \pmb{w}_{jij} \right] \right\vert^2. 
\end{equation} 
Therefore, one can directly apply \cite[Lemma~1]{8913624} to obtain the required lower bound.

\textbf{\textit{A-2) Proof of Theorem~1:}} We start by computing the variance of the desired signals in \eqref{eq:down:zf:b}, i.e., $\sum_{j \in \omega} \sqrt{\rho_{\rm dl} / \lambda_j^{\rm zf}} \left( \beta_{jil} / \beta_{jij} \right) s_j [i]$. It is readily verified that this is $\sum_{j \in \omega} \left( \rho_{\rm dl} / \lambda_j^{\rm zf} \right)  \left( \beta_{jil} / \beta_{jij}  \right)^2$, since $\mathbb{E} [ \left\vert s_j [k] \right\vert^2 ] =1, \forall j, k$. For the variance of the additive noise, $z_{il}^{\prime \prime}$, we have $\textrm{var} \left[ z_{il}^{\prime \prime} \right] = \textrm{var} \left[ \; \textrm{Interference due to estimation error} \; \right] + \textrm{var} \left[ z_{il} \right]$, which is due to the fact that $z_{il}$ is uncorrelated from the interference caused by the estimation error. For the first variance we obtain
\begin{align}
\nonumber &\textrm{var} \left[  \textrm{Interference due to estimation error}  \right] \\ 
&\hspace{5mm}=\sum_{j=1}^L \left( \frac{\rho_{\rm dl} }{\lambda_j^{\rm zf}} \right)  \sum_{k=1}^K  \mathbb{E} \left[ \left\vert  \pmb{\epsilon}_{jil}^{\dagger} \pmb{w}_{jkj} \right\vert^2 \right] \label{app:a:8:3}\\
&\hspace{5mm}= \sum_{j=1}^L \left( \frac{\rho_{\rm dl} }{\lambda_j^{\rm zf}} \right)  \sum_{k=1}^K \mathbb{E}  \left[ \textrm{tr} \left(  \pmb{\epsilon}_{jil}^{\dagger} \pmb{w}_{jkj}     \pmb{w}_{jkj}^{\dagger} \pmb{\epsilon}_{jil}   \right) \right] \\
&\hspace{5mm}= \sum_{j=1}^L \left( \frac{\rho_{\rm dl} }{\lambda_j^{\rm zf}} \right)  \sum_{k=1}^K \textrm{tr}  \left( \mathbb{E} \left[ \pmb{\epsilon}_{jil} \pmb{\epsilon}_{jil}^{\dagger} \right] \mathbb{E} \left[  \pmb{w}_{jkj} \pmb{w}_{jkj}^{\dagger}  \right] \right) \\
\nonumber &\hspace{5mm}=  \sum_{j=1}^L \left( \frac{\rho_{\rm dl} }{\lambda_j^{\rm zf}} \right)  \sum_{k=1}^K  \left( \beta_{jil} -  \sqrt{\rho_{\rm p}} \beta_{jil} \alpha_{jil} \right) \\
\nonumber &\hspace{32mm}\times\frac{1}{\left( M - K \right) \sqrt{\rho_{\rm p}} \beta_{jkj} \alpha_{jkj}},
\end{align} 
where in the last step we have used the following standard result in random matrix theory \cite{tulino2004random} 
\begin{equation}\label{app:a:8:4}
\mathbb{E} \left[  \pmb{w}_{jkj}^{\dagger} \pmb{w}_{jkj}  \right] = \frac{1}{\left( M - K \right) \sqrt{\rho_{\rm p}} \beta_{jkj} \alpha_{jkj}}.
\end{equation}
Also, for the variance of the zero-mean Gaussian noise $z_{il}$, we have $\mathbb{E} [ \left\vert z_{il} \right\vert^2 ] =1$. Lastly, since $\lbrace s_j[i] \rbrace_{j \in \Omega}$ are the only signals that are decoded jointly, the remaining signals in the first sum in \eqref{eq:down:zf:c} will be treated as noise. Thus, to compute the variance of the effective noise, one should further add $\sum_{j \in \Omega^c} \left( \rho_{\rm dl} / \lambda_j^{\rm zf} \right)  \left( \beta_{jil} / \beta_{jij}  \right)^2$ to the denominator in \eqref{thm:down:zf:1}. Therefore, \cite[Lemma~1]{8913624} can be directly applied to obtain the required lower bound. 

\textbf{\textit{A-3) An achievable region for RS:}} Following a technique used in \cite{bandemer2015optimal}, we provide analysis of the probability of error for the proposed RS scheme, when applied to the case of $L=2$. First, note that after dropping the index $i$, the cloud center and the satellite codeword generated at BS $l$, $l=1, 2,$ are given by $\pmb{s}_l^{(b)} ( m_l^{(b)} )$, and $\pmb{s}_l ( m_l^{(a)}, m_l^{(b)} )$, respectively. We only show the achievability proof at receiver $1$, i.e., user of cell $1$, as a similar analysis can be applied at receiver $2$, i.e., user of cell $2$.

Receiver $1$ tries to \textit{uniquely} recover both parts of its intended signal's message, $(m_1^{(a)}, m_1^{(b)})$ and to \textit{non-uniquely} recover messages from each layer of the interfering signal, $(m_2^{(a)}, m_2^{(b)})$. Therefore, receiver $1$ finds the unique pair $(\hat{m}_1^{(a)}, \hat{m}_1^{(b)})$ such that
\begin{align}\label{eq:app:rs:1}
&\left( \pmb{s}_1^{(b)} ( \hat{m}_1^{(b)} ), \pmb{s}_1 ( \hat{m}_1^{(a)}, \hat{m}_1^{(b)} ), \pmb{s}_2^{(b)} ( m_2^{(b)} ), \pmb{s}_2 ( m_2^{(a)}, m_2^{(b)} ), \pmb{y}_1 \right) \nonumber \\
&\qquad \qquad \qquad \in \mathcal{T}_{\epsilon}^n, \; \textrm{for some} \; ( m_2^{(a)}, m_2^{(b)} ),
\end{align}  
where $\mathcal{T}_{\epsilon}^n$ is the set of $\epsilon$-typical $n$-sequences (see \cite[Section 2.4]{el2011network} for definition of typical sets). 

Assume without loss of generality that the message pairs $(m_1^{(a)}, m_1^{(b)}) = (1, 1)$ and $(m_2^{(a)}, m_2^{(b)}) = (1, 1)$ are sent. Receiver $1$ declares an error if one or both of the following error events happen:
\begin{align}\label{eq:app:rs:2} 
E_1 &= \left\lbrace \left(  \pmb{s}_1^{(b)} ( 1 ), \pmb{s}_1 ( 1, 1 ), \pmb{s}_2^{(b)} ( 1 ), \pmb{s}_2 ( 1, 1 ), \pmb{y}_1 \right) \notin \mathcal{T}_{\epsilon}^n \right\rbrace \\
E_2 &= \Big\lbrace  \Big( \pmb{s}_1^{(b)} ( m_1^{(b)} ), \pmb{s}_1 ( m_1^{(a)}, m_1^{(b)} ), \pmb{s}_2^{(b)} ( m_2^{(b)} ), \nonumber \\ 
&\qquad \qquad \qquad \qquad \quad \pmb{s}_2 ( m_2^{(a)}, m_2^{(b)} ), \pmb{y}_1 \Big) \in \mathcal{T}_{\epsilon}^n, \nonumber \\
&\qquad \qquad \qquad \qquad \quad \textrm{for some} \; (m_1^{(a)}, m_1^{(b)}) \neq (1, 1), \; \nonumber\\ 
&\qquad \qquad \qquad \qquad \quad \textrm{and some} \; (m_2^{(a)}, m_2^{(b)})  \Big\rbrace.
\end{align}
By the law of large numbers, $P(E_1) \rightarrow 0$, as $n \rightarrow \infty$. We bound $P(E_2)$ in three different ways. As in \cite{bandemer2015optimal}, note that the joint typicality of the tuple 
$( \pmb{s}_1^{(b)} ( m_1^{(b)} ), \pmb{s}_1 ( m_1^{(a)}, m_1^{(b)} ), \pmb{s}_2^{(b)} ( m_2^{(b)} ), \pmb{s}_2 ( m_2^{(a)}, m_2^{(b)} ), \pmb{y}_1 )$ 
implies that $( \pmb{s}_1^{(b)} ( m_1^{(b)} ), \pmb{s}_1 ( m_1^{(a)}, m_1^{(b)} ), \pmb{y}_1 ) \in \mathcal{T}_{\epsilon}^n$, 
i.e., the triple 
$( \pmb{s}_1^{(b)} ( m_1^{(b)} ), \pmb{s}_1 ( m_1^{(a)}, m_1^{(b)} ), \pmb{y}_1 )$ 
is jointly typical. Hence, 
$E_2 \subseteq \lbrace  ( \pmb{s}_1^{(b)} ( m_1^{(b)} ), \pmb{s}_1 ( m_1^{(a)}, m_1^{(b)} ), \pmb{y}_1 )$ $\in \mathcal{T}_{\epsilon}^n,$  for some $\; (m_1^{(a)}, m_1^{(b)}) \neq (1, 1)  \rbrace = E_{21}$. 
Note that $E_{21}$ can be partitioned into the following 3 events:
\begin{align}
E_{21}^{(1)} &= \Big\lbrace \left( \pmb{s}_1^{(b)} ( m_1^{(b)} ), \pmb{s}_1 ( 1, m_1^{(b)} ), \pmb{y}_1 \right) \in \mathcal{T}_{\epsilon}^n, \nonumber \\
&\qquad\qquad\qquad\quad \textrm{for some} \; m_1^{(b)} \neq 1 \Big\rbrace, \\
E_{21}^{(2)} &= \Big\lbrace \left( \pmb{s}_1^{(b)} ( 1 ), \pmb{s}_1 ( m_1^{(a)}, 1 ), \pmb{y}_1 \right) \in \mathcal{T}_{\epsilon}^n, \; \nonumber \\
&\qquad\qquad\qquad\quad \textrm{for some} \; m_1^{(a)} \neq 1 \Big\rbrace, \\
E_{21}^{(3)} &= \Big\lbrace \left( \pmb{s}_1^{(b)} ( m_1{(b)} ), \pmb{s}_1 ( m_1^{(a)}, m_1^{(b)} ), \pmb{y}_1 \right) \in \mathcal{T}_{\epsilon}^n, \; \nonumber \\
&\qquad\qquad\qquad\quad
\textrm{for some} \; (m_1^{(a)}, m_1^{(b)}) \neq (1, 1) \Big\rbrace,
\end{align}
leading to $P(E_{21}) \leq \sum_{t=1}^3 P(E_{21}^{(t)})$. By the packing lemma, \cite[Section 3.2]{el2011network}, $P(E_{21}^{(1)})$ through $P(E_{21}^{(3)})$ tend to zero, as $n \rightarrow \infty$, if the following constraints are satisfied
\begin{align}\label{eq:app:rs:5}
R_1^{\rm (b)} &\leq I \left( s_1, s_1^{(b)} ; y_1 \right)  \\
R_1^{\rm (a)} &\leq  I \left( s_1 ; y_1 \vert s_1^{(b)} \right) \label{app:rs:5:1} \\
R_1^{\rm (a)} + R_1^{\rm (b)} &\leq  I  \left( s_1, s_1^{(b)} ; y_1  \right). \label{app:rs:5:2}
\end{align}
Notice that due to the codewords construction, the r.h.s in \eqref{eq:app:rs:5} and \eqref{app:rs:5:2} are identical, however the former is not necessary since the latter is the tighter condition. Therefore, we are left only with two rate constraints, \eqref{app:rs:5:1} and \eqref{app:rs:5:2}. Further note that the special structure of the codewords yields $I  ( s_1, s_1^{(b)} ; y_1  ) = I  \left( s_1 ; y_1  \right)$. 

\sloppy In addition, note that the joint typicality of the tuple $( \pmb{s}_1^{(b)} ( m_1^{(b)} ),$ $\pmb{s}_1 ( m_1^{(a)}, m_1^{(b)} ),$ $\pmb{s}_2^{(b)} ( m_2^{(b)} ),$ $\pmb{s}_2 ( m_2^{(a)}, m_2^{(b)} ),$ $\pmb{y}_1 )$ implies that $( \pmb{s}_1^{(b)} ( m_1^{(b)} ),$ $\pmb{s}_1 ( m_1^{(a)}, m_1^{(b)} ),$ $\pmb{s}_2^{(b)} ( m_2^{(b)} ), \pmb{y}_1)$ $\in$ $\mathcal{T}_{\epsilon}^n$, i.e., the quadruple $( \pmb{s}_1^{(b)} ( m_1^{(b)} ),$ $\pmb{s}_1 ( m_1^{(a)}, m_1^{(b)} ),$ $\pmb{s}_2^{(b)} ( m_2^{(b)} ),$ $\pmb{y}_1 )$ is jointly typical. Consequently, $E_2$ $\subseteq$ $\lbrace  ( \pmb{s}_1^{(b)} ( m_1^{(b)} ),$ $ \pmb{s}_1 ( m_1^{(a)}, m_1^{(b)} ),$ $ \pmb{s}_2^{(b)} ( m_2^{(b)} ),$ $\pmb{y}_1 )$ $\in$ $\mathcal{T}_{\epsilon}^n,$ for some $(m_1^{(a)}, m_1^{(b)}) \neq (1, 1), \; \textrm{and some} \; m_2^{(b)}   \rbrace = E_{22}$. The event $E_{22}$ can be partitioned into the following 6 events:
\begin{align}
E_{22}^{(1)} &= \Big\lbrace \left( \pmb{s}_1^{(b)} ( m_1^{(b)} ), \pmb{s}_1 ( 1, m_1^{(b)} ), \pmb{s}_2^{(b)} ( 1 ), \pmb{y}_1 \right) \in \mathcal{T}_{\epsilon}^n, \nonumber \\ 
&\qquad\quad\;\; \textrm{for some} \; m_1^{(b)} \neq 1  \Big\rbrace,  \label{app:rs:7:1} \\ 
E_{22}^{(2)} &= \Big\lbrace  \left( \pmb{s}_1^{(b)} ( 1 ), \pmb{s}_1 ( m_1^{(a)}, 1 ), \pmb{s}_2^{(b)} ( 1 ), \pmb{y}_1 \right) \in \mathcal{T}_{\epsilon}^n, \nonumber \\
&\qquad\quad\;\; \textrm{for some} \; m_1^{(a)} \neq 1 \Big\rbrace \label{app:rs:7:2}\\
E_{22}^{(3)} &= \Big\lbrace \left( \pmb{s}_1^{(b)} ( m_1^{(b)} ), \pmb{s}_1 ( m_1^{(a)}, m_1^{(b)} ), \pmb{s}_2^{(b)} ( 1 ), \pmb{y}_1 \right) \in \mathcal{T}_{\epsilon}^n, \nonumber \\
&\qquad\quad\;\; \textrm{for some} \; (m_1^{(a)}, m_1^{(b)}) \neq (1, 1)  \Big\rbrace \label{app:rs:7:3}\\
E_{22}^{(4)} &= \Big\lbrace \left( \pmb{s}_1^{(b)} ( m_1^{(b)} ), \pmb{s}_1 ( 1, m_1^{(b)} ), \pmb{s}_2^{(b)} ( m_2^{(b)} ), \pmb{y}_1 \right) \in \mathcal{T}_{\epsilon}^n, \nonumber \\
&\qquad\quad\;\; \textrm{for some} \;   m_1^{(b)} \neq 1,  \textrm{and some} \; m_2^{(b)} \neq 1  \Big\rbrace \label{app:rs:7:4}\\
E_{22}^{(5)} &= \Big\lbrace \left( \pmb{s}_1^{(b)} ( 1 ), \pmb{s}_1 ( m_1^{(a)}, 1 ), \pmb{s}_2^{(b)} ( m_2^{(b)} ), \pmb{y}_1 \right) \in \mathcal{T}_{\epsilon}^n, \nonumber \\
&\qquad\quad\;\; \textrm{for some} \;  m_1^{(a)} \neq 1, \; \textrm{and some} \; m_2^{(b)} \neq 1   \Big\rbrace \label{app:rs:7:5}\\
E_{22}^{(6)} &= \Big\lbrace  \left( \pmb{s}_1^{(b)} ( m_1^{(b)} ), \pmb{s}_1 ( m_1^{(a)}, m_1^{(b)} ), \pmb{s}_2^{(b)} ( m_2^{(b)} ), \pmb{y}_1 \right) \in \mathcal{T}_{\epsilon}^n,  \nonumber \\
&\qquad\quad\;\; \textrm{for some} \; (m_1^{(a)}, m_1^{(b)}) \neq (1, 1), \nonumber \\
&\qquad\quad\;\; \textrm{and some} \; m_2^{(b)} \neq 1  \Big\rbrace, \label{app:rs:7:6}
\end{align}
leading to $P(E_{22}) \leq \sum_{t=1}^6 P(E_{22}^{(t)})$. By the packing lemma, the probabilities $P(E_{22}^{(1)})$ through $P(E_{22}^{(6)})$ all tend to zero, as $n \rightarrow \infty$, if the following constraints are satisfied
\begin{align}\label{eq:app:rs:9}
R_1^{\rm (b)} &\leq I \left( s_1, s_1^{(b)} ; y_1 \vert s_2^{(b)}  \right)  \\
R_1^{\rm (a)} &\leq  I \left( s_1 ; y_1 \vert s_1^{(b)}, s_2^{(b)}  \right) \label{app:rs:8:1}\\
R_1^{\rm (a)} + R_1^{\rm (b)} &\leq  I \left( s_1, s_1^{(b)} ; y_1 \vert s_2^{(b)}  \right) \label{app:rs:8:2}\\
R_1^{\rm (b)} + R_2^{\rm (b)} &\leq  I \left(  s_1, s_1^{(b)}, s_2^{(b)} ; y_1 \right) \label{app:rs:8:3}\\
R_1^{\rm (a)} + R_2^{\rm (b)} &\leq I \left( s_1, s_2^{(b)} ; y_1 \vert s_1^{(b)} \right) \label{app:rs:8:4}\\
R_1^{\rm (a)} + R_1^{\rm (b)} +  R_2^{\rm (b)}  &\leq  I \left(  s_1, s_1^{(b)}, s_2^{(b)} ; y_1 \right) . \label{app:rs:8:5}  
\end{align}
Notice that due to the codewords construction, the r.h.s of \eqref{eq:app:rs:9} and \eqref{app:rs:8:2} are identical, however the latter is the tighter condition and thus the former can be omitted. Similarly, it is verified that \eqref{app:rs:8:3} is not necessary, since the constraint of \eqref{app:rs:8:5} is the tighter condition. As such, by removing \eqref{eq:app:rs:9} and \eqref{app:rs:8:3}, we are left with only four necessary constraints, i.e., \eqref{app:rs:8:1}, \eqref{app:rs:8:2}, \eqref{app:rs:8:4} and \eqref{app:rs:8:5}. Also, due to the structure of the codewords, we have $I ( s_1, s_1^{(b)} ; y_1 \vert s_2^{(b)}  ) = I ( s_1 ; y_1 \vert s_2^{(b)}  )$ and $I (  s_1, s_1^{(b)}, s_2^{(b)} ; y_1 ) = I (  s_1, s_2^{(b)} ; y_1 )$.

Lastly, the third way to bound $P(E_2)$ is to partition $E_2$ into the following 12 events:
\begin{align*}
\nonumber E_2^{(1)} &= \Big\lbrace \Big( \pmb{s}_1^{(b)} ( m_1^{(b)} ), \pmb{s}_1 ( 1, m_1^{(b)} ), \pmb{s}_2^{(b)} ( 1 ), \\
&\qquad\qquad \pmb{s}_2 ( 1, 1 ), \pmb{y}_1 \Big) \in \mathcal{T}_{\epsilon}^n, \textrm{for some} \; m_1^{(b)} \neq 1 \Big\rbrace, \\
\nonumber E_2^{(2)} &= \Big\lbrace \Big( \pmb{s}_1^{(b)} ( 1 ), \pmb{s}_1 ( m_1^{(a)}, 1 ), \pmb{s}_2^{(b)} ( 1 ), \\
&\qquad\qquad  \pmb{s}_2 ( 1, 1 ), \pmb{y}_1 \Big) \in \mathcal{T}_{\epsilon}^n, \textrm{for some} \; m_1^{(a)} \neq 1 \Big\rbrace,  \\
\nonumber E_2^{(3)} &= \Big\lbrace \Big( \pmb{s}_1^{(b)} ( m_1^{(b)} ), \pmb{s}_1 ( m_1^{(a)}, m_1^{(b)} ),  \\
&\qquad\qquad \pmb{s}_2^{(b)} ( 1 ), \pmb{s}_2 ( 1, 1 ), \pmb{y}_1 \Big) \in \mathcal{T}_{\epsilon}^n, \nonumber \\
&\qquad\qquad \textrm{for some} \; (m_1^{(a)}, m_1^{(b)}) \neq (1,1) \Big\rbrace, \\
\nonumber E_2^{(4)} &= \Big\lbrace \Big( \pmb{s}_1^{(b)} ( m_1^{(b)} ), \pmb{s}_1 ( 1, m_1^{(b)} ), \\
&\qquad\qquad \pmb{s}_2^{(b)} ( m_2^{(b)} ), \pmb{s}_2 ( 1, m_2^{(b)} ), \pmb{y}_1 \Big) \in \mathcal{T}_{\epsilon}^n, \nonumber \\
&\qquad\qquad \textrm{for some} \; m_1^{(b)} \neq 1,  \textrm{and some} \; m_2^{(b)} \neq 1 \Big\rbrace, \\
\nonumber E_2^{(5)} &= \Big\lbrace \Big( \pmb{s}_1^{(b)} ( 1 ), \pmb{s}_1 ( m_1^{(a)}, 1 ), \\
&\qquad\qquad \pmb{s}_2^{(b)} ( m_2^{(b)} ), \pmb{s}_2 ( 1, m_2^{(b)} ), \pmb{y}_1 \Big) \in \mathcal{T}_{\epsilon}^n,  \nonumber \\
&\qquad\qquad \textrm{for some} \; m_1^{(a)}  \neq 1,  \textrm{and some} \; m_2^{(b)} \neq  1 \Big\rbrace, \\
\nonumber E_2^{(6)} &= \Big\lbrace \Big( \pmb{s}_1^{(b)} ( m_1^{(b)} ), \pmb{s}_1 ( m_1^{(a)}, m_1^{(b)} ), \\
&\qquad\qquad  \pmb{s}_2^{(b)} ( m_2^{(b)} ), \pmb{s}_2 ( 1, m_2^{(b)} ), \pmb{y}_1 \Big) \in \mathcal{T}_{\epsilon}^n, \nonumber \\
&\qquad\qquad \textrm{for some} \; (m_1^{(a)}, m_1^{(b)} ) \neq (1, 1), \nonumber \nonumber \\
&\qquad\qquad  \textrm{and some} \; m_2^{(b)}  \neq 1 \Big\rbrace, \\
\nonumber E_2^{(7)} &=  \Big\lbrace \Big( \pmb{s}_1^{(b)} ( m_1^{(b)} ), \pmb{s}_1 ( 1, m_1^{(b)} ), \\
&\qquad\qquad \pmb{s}_2^{(b)} ( 1 ), \pmb{s}_2 ( m_2^{(a)}, 1 ), \pmb{y}_1 \Big) \in \mathcal{T}_{\epsilon}^n,  \nonumber \\
&\qquad\qquad \textrm{for some} \; m_1^{(b)} \neq 1,  \textrm{and some} \; m_2^{(a)}  \neq 1 \Big\rbrace, \\
\nonumber E_2^{(8)} &= \Big\lbrace \Big( \pmb{s}_1^{(b)} ( 1 ), \pmb{s}_1 ( m_1^{(a)}, 1 ), \\
&\qquad\qquad \pmb{s}_2^{(b)} ( 1 ), \pmb{s}_2 ( m_2^{(a)}, 1 ), \pmb{y}_1 \Big) \in \mathcal{T}_{\epsilon}^n, \nonumber \\
&\qquad\qquad \textrm{for some} \; m_1^{(a)} \neq 1,  \textrm{and some} \; m_2^{(a)}  \neq 1 \Big\rbrace, \\
\nonumber E_2^{(9)} &= \Big\lbrace \Big( \pmb{s}_1^{(b)} ( m_1^{(b)} ), \pmb{s}_1 ( m_1^{(a)}, m_1^{(b)} ), \\
&\qquad\qquad \pmb{s}_2^{(b)} ( 1 ), \pmb{s}_2 ( m_2^{(a)}, 1 ), \pmb{y}_1 \Big) \in \mathcal{T}_{\epsilon}^n, \nonumber \\
&\qquad\qquad
\textrm{for some} \; (m_1^{(a)}, m_1^{(b)}) \neq (1, 1), \nonumber \\
&\qquad\qquad \textrm{and some} \; m_2^{(a)}  \neq 1 \Big\rbrace, \\
\nonumber E_2^{(10)} &= \Big\lbrace \Big( \pmb{s}_1^{(b)} ( m_1^{(b)} ), \pmb{s}_1 (1, m_1^{(b)} ), \\
&\qquad\qquad \pmb{s}_2^{(b)} (  m_2^{(b)} ), \pmb{s}_2 ( m_2^{(a)},  m_2^{(a)} ), \pmb{y}_1 \Big) \in \mathcal{T}_{\epsilon}^n, \nonumber \\
&\qquad\qquad \textrm{for some} \;  m_1^{(b)} \neq 1, \nonumber \\
&\qquad\qquad \textrm{and some} \; (m_2^{(a)}, m_2^{(b)}) \neq (1, 1) \Big\rbrace, \\
\nonumber E_2^{(11)} &= \Big\lbrace \Big( \pmb{s}_1^{(b)} ( 1 ), \pmb{s}_1 (m_1^{(a)}, 1 ), \\
&\qquad\qquad \pmb{s}_2^{(b)} (  m_2^{(b)} ), \pmb{s}_2 ( m_2^{(a)},  m_2^{(a)} ), \pmb{y}_1 \Big) \in \mathcal{T}_{\epsilon}^n, \nonumber \\
&\qquad\qquad \textrm{for some} \;  m_1^{(a)} \neq 1, \nonumber  \\
&\qquad\qquad \textrm{and some} \; (m_2^{(a)}, m_2^{(b)}) \neq (1, 1) \Big\rbrace, \\
\nonumber E_2^{(12)} &= \Big\lbrace \Big( \pmb{s}_1^{(b)} ( m_1^{(b)} ), \pmb{s}_1 (m_1^{(a)}, m_1^{(b)} ), \\
&\qquad\qquad  \pmb{s}_2^{(b)} (  m_2^{(b)} ), \pmb{s}_2 ( m_2^{(a)},  m_2^{(a)} ), \pmb{y}_1 \Big) \in \mathcal{T}_{\epsilon}^n, \nonumber \\
&\qquad\qquad \textrm{for some} \;  ( m_1^{(a)}, m_1^{(b)}) \neq (1, 1), \nonumber \\
&\qquad\qquad \textrm{and some} \; (m_2^{(a)}, m_2^{(b)}) \neq (1, 1) \Big\rbrace, 
\end{align*}
leading to $P(E_2) \leq \sum_{t=1}^{12} P(E_2^{(t)}) $. Using the packing lemma, the probabilities $P(E_2^{(1)})$ through $P(E_2^{(12)})$ above tend to zero, as $n \rightarrow \infty$, if the following conditions are satisfied 
\begin{align}\label{eq:app:rs:12}
R_1^{\rm (b)} &\leq I \left( s_1, s_1^{(b)} ; y_1 \vert s_2, s_2^{(b)} \right)  \\
R_1^{\rm (a)} &\leq  I \left( s_1 ; y_1 \vert s_1^{(b)}, s_2, s_2^{(b)} \right) \label{app:rs:12:1}\\
R_1^{\rm (a)} + R_1^{\rm (b)} &\leq  I \left( s_1, s_1^{(b)} ; y_1 \vert s_2, s_2^{(b)} \right) \label{app:rs:12:2}\\
R_1^{\rm (b)} + R_2^{\rm (b)} &\leq  I \left( s_1, s_1^{(b)}, s_2, s_2^{(b)} ; y_1 \right) \label{app:rs:12:3}\\
R_1^{\rm (a)} + R_2^{\rm (b)} &\leq  I \left( s_1, s_2, s_2^{(b)} ; y_1 \vert s_1^{(b)} \right) \label{app:rs:12:4}\\
R_1^{\rm (a)} + R_1^{\rm (b)} + R_2^{\rm (b)} &\leq  I \left( s_1, s_1^{(b)}, s_2, s_2^{(b)} ; y_1 \right) \label{app:rs:12:5}\\
R_1^{\rm (b)} + R_2^{\rm (a)} &\leq  I \left( s_1, s_1^{(b)}, s_2 ; y_1 \vert s_2^{(b)} \right) \label{app:rs:12:6}\\
R_1^{\rm (a)} + R_2^{\rm (a)} &\leq  I \left( s_1, s_2 ; y_1 \vert s_1^{(b)}, s_2^{(b)} \right) \label{app:rs:12:7}\\
R_1^{\rm (a)} + R_1^{\rm (b)} + R_2^{\rm (a)} &\leq I \left( s_1, s_1^{(b)}, s_2 ; y_1 \vert s_2^{(b)} \right)  \label{app:rs:12:8}\\
R_1^{\rm (b)} + R_2^{\rm (a)} + R_2^{\rm (b)} &\leq  I \left( s_1, s_1^{(b)}, s_2, s_2^{(b)} ; y_1 \right) \label{app:rs:12:9}\\
R_1^{\rm (a)} + R_2^{\rm (a)} + R_2^{\rm (b)} &\leq  I \left( s_1, s_2, s_2^{(b)} ; y_1 \vert s_1^{(b)} \right) \label{app:rs:12:10}\\
R_1^{\rm (a)} + R_1^{\rm (b)} + R_2^{\rm (a)} + R_2^{\rm (b)} &\leq  I \left( s_1, s_1^{(b)}, s_2, s_2^{(b)} ; y_1 \right). \label{app:rs:12:11} 
\end{align}
It can be readily seen that the constraints of \eqref{eq:app:rs:12}, \eqref{app:rs:12:3}, \eqref{app:rs:12:4}, \eqref{app:rs:12:5}, \eqref{app:rs:12:6} and \eqref{app:rs:12:9} are redundant and can be removed. We are thus left with only 6 necessary constraints, i.e., \eqref{app:rs:12:1}, \eqref{app:rs:12:2}, \eqref{app:rs:12:7}, \eqref{app:rs:12:8}, \eqref{app:rs:12:10} and \eqref{app:rs:12:11}. Also, the code construction yields $I ( s_1, s_1^{(b)} ; y_1 \vert s_2, s_2^{(b)} )$ $=$ $I ( s_1 ; y_1 \vert s_2 )$, $I ( s_1 ; y_1 \vert s_1^{(b)}, s_2, s_2^{(b)} )$ $=$ $I ( s_1 ; y_1 \vert s_1^{(b)}, s_2 )$, $I ( s_1, s_1^{(b)}, s_2, s_2^{(b)} ; y_1 )$ $=$ $I ( s_1, s_2 ; y_1 )$, $I ( s_1, s_2, s_2^{(b)} ; y_1 \vert s_1^{(b)} )=I ( s_1, s_2 ; y_1 \vert s_1^{(b)} ),$ and $I ( s_1, s_1^{(b)}, s_2 ; y_1 \vert s_2^{(b)} ) = I ( s_1, s_2 ; y_1 \vert s_2^{(b)} )$.
 
Hence, by bounding $P(E_2)$ using these three different approaches, an achievable region is established at receiver $1$ (denoted by $\mathcal{R}_1^{\rm RS}$), which is the union of the three regions described above. One can similarly obtain the achievable region at receiver $2$ (denoted by $\mathcal{R}_2^{\rm RS}$) by replacing $y_1$ with $y_2$ and swapping appropriate indices. The network-wide achievable region obtained by the generalized RS scheme in conjunction with \textit{non-unique} decoding for a two-cell system can then be written in the following form 
\begin{equation}\label{app:rs:a}
\mathscr{R}^{\rm RS} = \bigcap_{l=1}^2 \mathcal{R}_{l}^{\rm RS}, 
\end{equation} 
where 
\begin{equation}\label{app:rs:b}
\mathcal{R}_{l}^{\rm RS} = \bigcup_{ \Omega_l \in \mathcal{S}_l } \mathcal{R}_{\textrm{MAC}(\Omega_l, l ) }^{\rm RS},  \quad l=1, 2, 
\end{equation}
and $\mathcal{R}_{\textrm{MAC}(\Omega_l, l ) }^{\rm RS}$ is a modified MAC region, which has less than $2^{\vert \Omega_l \vert} -1 $ rate constraints, as some of the constraints are removed from the regular MAC region (due to the codewords construction as explained above), and $\mathcal{S}_l, l=1, 2,$ are given by 
\begin{align}
\nonumber \mathcal{S}_1^{\rm RS}  &= \Big\lbrace \left\lbrace m_1^{(a)}, m_1^{(b)} \right\rbrace, \left\lbrace m_1^{(a)}, m_1^{(b)}, m_2^{(b)} \right\rbrace, \\
 &\hspace{29mm}\left\lbrace m_1^{(a)}, m_1^{(b)}, m_2^{(a)}, m_2^{(b)} \right\rbrace  \Big\rbrace, \\
\nonumber \mathcal{S}_2^{\rm RS}  &= \Big\lbrace \left\lbrace m_2^{(a)}, m_2^{(b)} \right\rbrace, \left\lbrace m_2^{(a)}, m_2^{(b)}, m_1^{(b)} \right\rbrace, \\
&\hspace{29mm}\left\lbrace m_2^{(a)}, m_2^{(b)}, m_1^{(a)}, m_1^{(b)} \right\rbrace  \Big\rbrace. 
\end{align}
Specifically, if $\Omega_l$ contains messages of both layers $(m_{j}^{(a)},m_{j}^{(b)}),$ for some $j$, then those constraints that involve $R_{j}^{\rm  (b)}$ but not $R_{j}^{\rm  (a)}$ are not needed and will thus be removed from the rate region. In particular, the constraints that are removed from each of the three regions described above are: \eqref{eq:app:rs:5} from the $1^{\rm st}$ region, \eqref{eq:app:rs:9} and \eqref{app:rs:8:3} from the $2^{\rm nd}$ region, \eqref{eq:app:rs:12}, \eqref{app:rs:12:3}, \eqref{app:rs:12:4}, \eqref{app:rs:12:5}, \eqref{app:rs:12:6} and \eqref{app:rs:12:9} from the $3^{\rm rd}$ region.

One can also re-write $\mathcal{S}_l^{\rm RS}$ as follows 
\begin{align}
\nonumber \mathcal{S}_l^{\rm RS} = &\Big\lbrace \left\lbrace m_l^{(a)}, m_l^{(b)} \right\rbrace \Big\rbrace \\ 
&\times \left\lbrace \emptyset , \left\lbrace m_j^{(b)} \right\rbrace , \left\lbrace m_j^{(a)} , m_j^{(b)} \right\rbrace  \right\rbrace, \quad l=1, 2, \quad j \neq l. \label{app:rs:d}
\end{align}

\end{document}